\documentclass[noinfoline]{imsart}
\RequirePackage[numbers]{natbib}

\usepackage[utf8]{inputenc}
\usepackage{amsmath}
\usepackage{mathrsfs}

\usepackage{mathabx}

\usepackage{amsthm}

\usepackage{amsfonts}
\usepackage{amssymb}
\usepackage{graphicx}
\usepackage{placeins}
\usepackage{algorithm2e}
\usepackage{caption}
\usepackage{array}
\usepackage{multicol}
\usepackage{color}
\usepackage{url}
\usepackage{bbm}
\usepackage[font={small,it}]{caption}
\usepackage{stackrel}
\usepackage{enumerate}
\usepackage{datetime}
\usepackage{mhequ}
\usepackage{xfrac}

\startlocaldefs
\def \be{\begin{equs}}
\def \ee{\end{equs}}

\DeclareMathOperator{\sech}{sech}

\DeclareMathOperator{\rank}{rank}
\theoremstyle{plain}

\newtheorem{theorem}{Theorem}[section]

\newtheorem{assumption}[theorem]{Assumption}

\newtheorem{lemma}{Lemma}[section]

\theoremstyle{definition}
\newtheorem{remark}{Remark}[section]

\renewcommand{\P}{\mathcal P}
\newcommand{\X}{\mathbf X}
\newcommand{\R}{\mathbf R}
\newcommand{\D}{\mathrm{d}}
\newcommand{\E}{\mathbf E}
\DeclareMathOperator{\cov}{cov}
\DeclareMathOperator{\DD}{D}

\newcommand{\mc}[1]{\mathcal{#1}}

\newcommand{\Gam}[2]{\textnormal{Gamma}\left( #1,#2\right)}

\newcommand{\tv}[2]{\left| \left| #1 - #2 \right|\right|_{\textnormal{TV}}}
\newcommand{\gp}{\textnormal{GP}}

\newcommand{\norm}[1]{\left\| #1 \right\|}

\newcommand{\tSigma}{\Sigma_{\epsilon}}

\newcommand{\Bern}[1]{\textnormal{Bernoulli}\left(#1 \right)}

\newcommand{\No}[2]{\textnormal{Normal}\left(#1,#2 \right)}
\newcommand{\PG}[2]{\textnormal{PG}\left(#1,#2 \right)}

\newcommand{\KL}[2]{\textnormal{KL}\left( #1 ~||~ #2 \right)}
\newcommand{\tr}[1]{\textnormal{tr}\left( #1 \right)}
\newcommand{\bigO}[1]{\mc O\left( #1\right)}

\newcommand{\Log}[1]{\log \left( #1 \right)}

\newcommand{\Pp}{\mathbf P}

\DeclareMathOperator{\Dsc}{Disc}
\DeclareMathOperator{\Binom}{Binomial}
\newcommand{\Disc}[2]{\Dsc \left(#1, #2\right)}
\newcommand{\Kr}[2]{\mc P\left(#1,#2 \right)}
\newcommand{\Kre}[2]{\mc P_{\epsilon}\left(#1,#2 \right)}

\newcommand{\Krth}{\Kr{\theta}{\cdot}}
\newcommand{\Krthe}{\Kre{\theta}{\cdot}}

\newcommand{\Ueps}{U_{\epsilon}}
\newcommand{\Lameps}{\Lambda_{\epsilon}}
\newcommand{\Psieps}{\Psi_{\epsilon}}

\endlocaldefs
\graphicspath{{./}}

\makeatletter
\def\input@path{{./}}
\makeatother

\begin{document}
\allowdisplaybreaks
% paper-specific macros go here
\begin{frontmatter}

\title{Optimal approximating Markov chains for Bayesian inference}
\runtitle{Optimal Approximate MCMC}
\date{\today}

\begin{aug}
\author{\fnms{James E.} 
\snm{Johndrow}\thanksref{t2}\ead[label=e1]{jj@stat.duke.edu}},
\author{\fnms{Jonathan C.} 
\snm{Mattingly}\thanksref{}\ead[label=e2]{jonm@math.duke.edu}},
\author{\fnms{Sayan} 
\snm{Mukherjee}\thanksref{}\ead[label=e3]{sayan@stat.duke.edu}},
\and
\author{\fnms{David B.} \snm{Dunson}\thanksref{t2} 
\ead[label=e4]{dunson@duke.edu}
\ead[label=u1,url]{http://www.isds.duke.edu/\textasciitilde dunson/}}

%\thankstext{t1}{Some comment}
\thankstext{t2}{This research was partially supported by the National Science 
Foundation grant number 1546130.}
%\thankstext{t3}{Second supporter of the project}
\runauthor{J. Johndrow et al.}

\affiliation{Duke University\thanksmark{m1}}

\address{114 Old Chemistry Building\\
Duke University\\ 
Durham, NC 27708 \\
\printead{e1}\\
\printead{e2}\\
\printead{e3}\\
\printead{e4}\\
\printead{u1}} 

\end{aug}

\begin{abstract}
The Markov Chain Monte Carlo method is the dominant paradigm for posterior 
computation in Bayesian analysis. It is common to control computation time by 
making approximations to the Markov transition kernel. Comparatively little 
attention has been paid to computational optimality in these approximating 
Markov Chains, or when such approximations are justified relative to obtaining 
shorter paths from the exact kernel. We give simple, sharp bounds for uniform 
approximations of uniformly mixing Markov chains. We then suggest a notion of 
optimality that incorporates computation time and approximation error, and use 
our bounds to make generalizations about properties of good approximations in 
the uniformly mixing setting. The relevance of these properties is demonstrated 
in applications to a minibatching-based approximate MCMC algorithm for 
large $n$ logistic regression and low-rank approximations for Gaussian 
processes. 
\end{abstract}

\begin{keyword}
%\kwd[Primary ]{}
\kwd{approximate MCMC; Bayesian; big data; computation; Markov chain Monte 
Carlo; minibatch; perturbation theory}
%\kwd[; secondary ]{}
\end{keyword}

\end{frontmatter}

\section{Introduction} \label{sec:intro}
The fundamental entity in Bayesian statistics is the posterior distribution
\begin{align}
 \mu(x \mid w) = \frac{p(w \mid x) p(x)}{\int_{x} p(w \mid x) p(x)}, \label{eq:post}
\end{align}
the conditional distribution of the model parameters $x$ given the data 
$w$.\footnote{We will abuse notation to allow $\mu$ to represent both a measure 
and a density.} The integral in the denominator of (\ref{eq:post}) is typically 
not available in closed form. A common approach constructs an ergodic Markov 
chain with invariant distribution $\mu$, and then collects finite sample paths 
$X_1,\ldots,X_t$ from the chain. Statistical inference then relies on 
Ces\'{a}ro averages $n^{-1} \sum_{k=0}^{n-1} f(X_k)$ for appropriate classes of 
$\mu$-measurable functions $f$, and other pathwise quantities. This is referred 
to as Markov Chain Monte Carlo \cite{robert2004monte,gamerman2006markov} or 
MCMC. 

We consider Markov chains that result from approximating the transition kernel 
$\P$ by another kernel $\P_\epsilon$ satisfying $\sup_{x \in \X} \| 
\P(x,\cdot)-\P_\epsilon(x,\cdot) \|_{TV} < \epsilon$, where $\X$ is the state 
space. We refer to this as ``approximate MCMC.'' The use of approximate kernels 
-- often without showing such an error bound -- is common practice in Bayesian 
analysis, and is usually computationally motivated, i.e. obtaining samples from 
$\P_\epsilon$ requires less computation than sampling from $\P$. Our focus is on 
providing simple error bounds for $\P_\epsilon$ when $\P$ satisfies a uniform 
mixing condition, and constructing a notion of computational optimality that 
allows analysis of the optimal $\epsilon$ for a fixed computational budget 
under these ``nice'' convergence conditions. 

While being arguably the dominant algorithm for Bayesian inference, MCMC is 
computationally demanding when either $p$ (the dimension of $\X$) or $N$ (the 
number of observations) is large. To more easily apply MCMC in these cases, it 
is common to approximate $\P$ with a kernel that is simpler or faster to sample 
from.  One example is inference for Gaussian process models, bypassing $O(N^3)$ 
matrix inversion through approximations 
\cite{banerjee2008gaussian,banerjee2012efficient,hughes2013dimension}. Some 
numerical results showing that accurate approximations can be achieved with 
low-rank approximations of the covariance matrix when the state space is 
discrete are given in \citet{johndrow2017coupling}. Another prevalent example is 
the use of Laplace or Gaussian approximations to conditional distributions. 
\citet{guhaniyogi2014bayesian} proposes an algorithm that replaces some sampling 
steps with point estimates. \citet{korattikara2014austerity} approximate 
Metropolis-Hastings acceptance decisions using subsets of the data.  It is also 
common to approximate intractable full conditionals by simpler distributions, 
with \citet{bhattacharya2010nonparametric} using a beta approximation, 
\citet{o2004bayesian} replacing the logistic with a $t$ distribution, and 
\citet{ritter92} discretizing.

A significant literature exists giving error bounds for approximations 
$\P_\epsilon$ via perturbation theory. This includes 
\cite{roberts1998convergence} and \cite{mitrophanov2005sensitivity}, which both 
give results for uniformly ergodic $\P$ similar to those given here. 
\citet{johndrow2017coupling} revisits this setting via coupling and decoupling. 
\citet{pillai2014ergodicity} and \citet{rudolf2015perturbation} give results 
primarily in the Wasserstein metric for $\P$ satisfying a Wasserstein ergodicity 
and Foster-Lyapunov condition. Other examples of perturbation theory applied to 
MCMC include \cite{alquier2014noisy, ferre2013regular, bardenet2014towards}. Our 
goal is to use results of this sort to study a notion of optimality 
incorporating statistical and computational factors. This gives insight into 
optimal algorithms under limited computational budgets. Our main contributions 
are to define a notion of optimality incorporating both statistical and computational
properties and apply it to the study of aMCMC algorithms. We demonstrate the 
relevance of this notion through an application to minibatching MCMC for 
logisitic regression and Metropolis-Hastings algorithms for Gaussian process 
models.

\section{Basic error bounds} \label{sec:ergodic}
This section provides some simple error bounds for approximate MCMC. We use 
conditions similar to \citet{johndrow2017coupling} but take a different approach 
to showing the relevant inequalities, arriving at different bounds in some 
cases. The proofs are very direct and elementary.

\subsection{Preliminaries}
Consider an increasing sequence $\mathcal F_t$ of $\sigma$-algebras over a 
probability space, and $\mathcal F_t$-measurable Markov processes $X_t, 
X_t^\epsilon$ on a Polish state space $\X$ evolving according to transition 
kernels $\P, \P_\epsilon$ respectively with $\P, \P_\epsilon : \X \times \X \to 
\R_+$, the positive half-line. In many applications, $\X = \R^p$. We assume $\P$ 
has invariant measure $\mu$, and are interested in bounds on
\be
\D_{TV}(\mu,\widehat \mu_n^\epsilon) = \left\| \mu - n^{-1} \sum_{k=0}^{n-1} 
\nu \P^k_\epsilon \right\|_{TV},\quad \Delta(\mu f, \widehat f_n^\epsilon  ) = 
\E \left( \mu f - \frac1n \sum_{k=0}^{n-1} f(X_k^\epsilon) \right)^2 ,
\ee
where the expectation is taken with respect to the law of 
$X_0^\epsilon,X_1^\epsilon,\ldots,X_{n-1}^\epsilon$. We use two basic 
conditions. The first is a form of Doeblin's condition.
\begin{assumption}[Doeblin] \label{ass:Doeblin}
 There exists a constant $0 < \alpha < 1$ such that
 \be
 \| \P(x,\cdot) - \P(y,\cdot) \|_{TV} < 1-\alpha
 \ee
 for every $(x,y) \in \X \times \X$. 
\end{assumption}
This condition guarantees that $\P$ has a unique invariant measure that it 
converges to exponentially. We also assume that the approximation $\P_\epsilon$ 
is uniform in total variation.
\begin{assumption}[Approximation error] \label{ass:UniformTV}
 There exists $0 < \epsilon < \alpha$ such that
 \be
 \| \P(x,\cdot) - \P_\epsilon(x,\cdot) \|_{TV} < \epsilon
 \ee
 for all $x \in \X$. 
\end{assumption}
If $\epsilon < \alpha/2$, these two conditions together guarantee that 
$\P_\epsilon$ also has a unique invariant measure that it converges to 
exponentially.
\begin{remark} \label{rem:Aepsilon}
 Assume that Assumptions \ref{ass:Doeblin} and \ref{ass:UniformTV} hold, and 
the constant $\alpha$ in Assumption \ref{ass:Doeblin} cannot be globally 
improved. Then if $\epsilon < \frac{\alpha}{2}$
 \be
 \| \P_\epsilon(x,\cdot) - \P_\epsilon(y,\cdot) \|_{TV} \le 1-\alpha_\epsilon < 1
 \ee
 where $\alpha_\epsilon \in (\alpha-2\epsilon, \alpha+2 \epsilon)$. 
\end{remark}
\begin{proof}
Observe that
\be
\| \P_\epsilon(x,\cdot) - \P_\epsilon(y,\cdot) \|_{TV} &\le \|\P_\epsilon(x,\cdot) - \P(x,\cdot) \|_{TV} + \|\P(x,\cdot)-\P(y,\cdot)\|_{TV} \\
&+  \|\P(y,\cdot) - \P_\epsilon(y,\cdot)\|_{TV} \\
&\le \epsilon + 1-\alpha + \epsilon = 1-\alpha+2\epsilon = 1-\alpha_\epsilon < 1.
\ee
This establishes the existence of some $\alpha_\epsilon > \alpha-2\epsilon>0$. 
Now to obtain an upper bound, notice that
\be
\| \P(x,\cdot) - \P(y,\cdot) \|_{TV} &\le \|\P(x,\cdot) - \P_\epsilon(x,\cdot) 
\|_{TV} + \|\P_\epsilon(x,\cdot)-\P_\epsilon(y,\cdot)\|_{TV} \\
&+  \|\P(y,\cdot) - \P_\epsilon(y,\cdot)\|_{TV} \\ 
&\le 2\epsilon + 1-\alpha_\epsilon 
\ee
so since $\alpha$ cannot be globally improved, we must have
\be
1-\alpha_\epsilon+2\epsilon &\ge 1-\alpha \\
\alpha_\epsilon &\le \alpha+2\epsilon
\ee
else there would exist a ``better'' constant than $\alpha$ in Assumption 
\ref{ass:Doeblin}. This gives the result.
\end{proof}

Finally, we obtain a simple bound on the closeness of the invariant measures.
\begin{theorem} \label{thm:Bias}
 Suppose Assumptions \ref{ass:Doeblin} and \ref{ass:UniformTV} are satisfied 
and $\epsilon < \frac{\alpha}{2}$. Then 
 \be
 \| \mu - \mu_\epsilon \|_{TV} \le \frac{\epsilon}{\alpha},
 \ee
 where $\mu_\epsilon$ is the invariant measure of $\P_\epsilon$. 
\end{theorem}
\begin{proof}
 Observe that
 \be
 \| \mu - \mu_\epsilon \|_{TV} \le \| \mu \P - \mu_\epsilon \P \|_{TV} + 
\|\mu_\epsilon \P - \mu_\epsilon \P_\epsilon \|_{TV} \le (1-\alpha) \| \mu - 
\mu_\epsilon \|_{TV} + \epsilon.
 \ee
 rearrangement produces the result.
\end{proof}

These conditions suffice to obtain bounds on $\D_{TV}$ and $\Delta$.

\subsection{Main Results}

\begin{theorem} \label {thm:approxerr}
 Suppose $\P$ satisfies Assumption \ref{ass:Doeblin}, $\P_{\epsilon}$ satisfies 
\ref{ass:UniformTV}. Let $X_0 \sim \nu$ for any probability measure $\nu$ on 
$(\X, \mathcal F_0)$. Then
 \be
  \left\| \mu - \frac1n \sum_{k=0}^{n-1} \nu \P^k_\epsilon \right\|_{TV} &\le 
\frac{(1-(1-\alpha)^n) \| \mu- \nu \|_{TV}}{n \alpha} - \frac{\epsilon 
(1-(1-\alpha)^n)}{n \alpha^2} + \frac{\epsilon}{\alpha}. 
\label{eq:tvboundapprox}
 \ee
\end{theorem}
\begin{proof}
By induction
\be
\nu  \P_\epsilon^n - \nu^*  \P^n = (\nu - \nu^*)  \P^n + \sum_{k=0}^{n-1} \nu  
\P_\epsilon^k ( \P_\epsilon -  \P)  \P^{n-k-1}.
\ee
We have that
\be
\| \nu  \P_\epsilon^k  \P_\epsilon - \nu  \P_\epsilon^k  \P \|_{TV} \le 
\epsilon
\ee
so
\be
\| \nu  \P_{\epsilon}^k  \P_\epsilon  \P^{n-k-1} - \nu  \P_{\epsilon}^k  \P  
\P^{n-k-1} \|_{TV} \le \epsilon (1-\alpha)^{n-k-1}.
\ee
Hence, by the triangle inequality
\be
\| \nu  \P_{\epsilon}^n - \nu^*  \P^n  \|_{TV} &\le \| \nu^*  \P^n - \nu  
\P^n\|_{TV} + \sum_{k=1}^{n-1} \| \nu  \P_{\epsilon}^k  \P_{\epsilon}  
\P^{n-k-1} - \nu  \P_{\epsilon}^k  \P  \P^{n-k-1} \|_{TV} \\
&\le (1-\alpha)^n \| \nu^* - \nu \|_{TV} + \epsilon \sum_{k=0}^{n-1} 
(1-\alpha)^{n-k-1} \\
&\le (1-\alpha)^n \| \nu^* - \nu \|_{TV} + \epsilon 
\frac{1-(1-\alpha)^n}{\alpha} \label{eq:PPeps}
\ee
Taking $\nu^* = \mu$ we have
\be
\left\| \mu - n^{-1} \sum_{k=0}^{n-1} \nu  \P^k_{\epsilon} \right\|_{TV} \le  
\frac{(1-(1-\alpha)^n) \| \nu - \mu \|_{TV}}{n \alpha} + \frac{\epsilon}{\alpha} 
+ \frac{\epsilon ((1-\alpha)^n-1)}{n \alpha^2}.
\ee
Completing the proof. A similar argument in the Wasserstein norm can be found 
in \cite{rudolf2015perturbation}.

\end{proof}
From \citet[\S 6]{johndrow2017coupling} this bound is sharp. We obtain bounds 
on $\Delta$ at stationarity directly when $\epsilon < \alpha/2$ via the 
existence of a Doeblin condition for $\P_\epsilon$ in this case. The result 
relies on the following simple lemma.

\begin{lemma} \label{lem:CovBound}
 Suppose $\P$ satisfies assumption \ref{ass:Doeblin}. Let $f$ and $g$ be 
bounded functions. Then with $X_0 \sim \mu$
\begin{align*}
\cov(f(X_j),g(X_k)) \le 2 (1-\alpha)^{|j-k|} |f|_* |g|_*,
\end{align*}
where $\|f\|_* = \inf_{c \in \R} \|f-c\|_{\infty}$. 
\end{lemma}

\begin{proof}
First suppose $\mu f = 0, \mu g = 0$ and $|f|_{\infty}, |g|_{\infty} < 1$. 
Without loss of generality, take $k \ge j$. Then for $X_0 \sim \mu$
\be
\cov(f(X_j),g(X_k)) &= \E[ \E[ f(X_j) g(X_k) \mid \mathcal F_j]] - \E[f(X_j)] 
\E[g(X_k)] \\
&= \E[ f(X_j) \E[ g(X_k) \mid \mathcal F_j]] \\
&\le \E[ f(X_j) (\P^{k-j} g)(X_j)  ] \le 2 (1-\alpha)^{k-j},
\ee
since $\sup_{|f| < 1} \E[ f(X_j) - \mu f \mid X_0=x ] = \sup_{|f| < 1} \E[ 
f(X_j) \mid X_0=x ] \le 2 (1-\alpha)^j$ and $\mu f = 0$. 

Now, since 
\be
 \cov(f(X_j),g(X_k)) = \cov(f(X_j)-c_1,g(X_k)-c_2)
\ee
for any $c_1, c_2 \in \R$, 
\be
\sup_{f,g : |f|<1,|g|<1} \cov(f(X_j),g(X_k)) &\le 2 (1-\alpha)^{k-j} 
\ee
Finally, since $\cov(f(X_j),g(X_k)) = \cov(g(X_k),f(X_j))$, we obtain
\be
\cov(f(X_j),g(X_k)) &\le 2 (1-\alpha)^{|j-k|} |f|_* |g|_* 
\ee
for any bounded $f,g$. 
\end{proof}

This gives a result for $\Delta$ starting from stationarity.

\begin{theorem} \label{thm:ExpStationary}
 Suppose $\P$ satisfies Assumption \ref{ass:Doeblin} and $\P_\epsilon$  
satisfies assumption \ref{ass:UniformTV}, $|f|_* < \infty$, and $X_0 \sim 
\mu_\epsilon$. Then  
 \be \label{eq:ExpectationStationary}
 \frac{1}{|f|_*^2} \E \left( \mu f - \frac1n \sum_{k=0}^{n-1} f(X_k^\epsilon) 
\right)^2 \le \frac{4 \epsilon^2}{\alpha^2} + 2 S(n,\alpha_\epsilon)
 \ee
 where
 \be
 S(n,\alpha_\epsilon) = \left(\frac{2}{\alpha_\epsilon n} + 
\frac{2}{\alpha_\epsilon n^2} + 
\frac{2(1-\alpha_\epsilon)^{n+1}}{\alpha_\epsilon^2 n^2} - \frac1n - 
\frac{2}{\alpha_\epsilon^2 n^2} \right).
 \ee
\end{theorem}

\begin{proof}
\be
\E  & \left( \mu f - \frac1n \sum_{k=0}^{n-1} \nu \P^k_{\epsilon} f  + \frac1n 
\sum_{k=0}^{n-1} \nu \P^k_{\epsilon} f - \frac1n \sum_{k=0}^{n-1} 
f(X_k^\epsilon) \right)^2 \\
= &  \left( \mu f - \frac1n \sum_{k=0}^{n-1} \nu \P^k_{\epsilon} f \right)^2  
+  \E \left( \frac1n \sum_{k=0}^{n-1} f(X_k^\epsilon) - \nu \P^k_\epsilon f 
\right)^2 \\
\le &  4|f|_*^2 \left( \frac{(1-(1-\alpha)^n) \| \mu- \mu_\epsilon \|_{TV}}{n 
\alpha}  - \frac{\epsilon (1-(1-\alpha)^n)}{n \alpha^2} + 
\frac{\epsilon}{\alpha} \right)^2 \\
+ & \frac{1}{n^2} \sum_{j=0}^{n-1} \sum_{k=0}^{n-1} \cov(f(X_j),f(X_k)) \\
\le & 4|f|_*^2 \left( \frac{(1-(1-\alpha)^n) \epsilon}{n \alpha^2}  - 
\frac{\epsilon (1-(1-\alpha)^n)}{n \alpha^2} + \frac{\epsilon}{\alpha} \right)^2 
+ \frac{2 |f|_*^2}{n^2} \sum_{j=0}^{n-1} \sum_{k=0}^{n-1} 
(1-\alpha_\epsilon)^{|j-k|} \\
\le & \frac{4 \epsilon^2 |f|_*^2}{\alpha^2} + 2 |f|_*^2 
\left(\frac{2}{\alpha_\epsilon n} + \frac{2}{\alpha_\epsilon n^2} + 
\frac{2(1-\alpha_\epsilon)^{n+1}}{\alpha_\epsilon^2 n^2} - \frac1n - 
\frac{2}{\alpha_\epsilon^2 n^2} \right) \\
\le & \frac{4 \epsilon^2 |f|_*^2}{\alpha^2} + 2 |f|_*^2 S(n,\alpha_\epsilon)
\ee
\end{proof}

It is also possible to obtain results for general starting measures.
\begin{theorem}
Suppose $\mc P$ satisfies Assumption \ref{ass:Doeblin}, $\mc P_{\epsilon}$ 
satisfies Assumption \ref{ass:UniformTV} with $\epsilon < \frac{\alpha}{2}$. Let 
$X_0 \sim \nu$ with $\nu \ll \mu_\epsilon$. Then for bounded $f$
\be
 \frac{1}{|f|_*^2} \Delta(\mu f, \widehat f_n^{\epsilon} ) &\le \left( \frac{2 
\epsilon}{\alpha} \right)^2 + 2 S( n, \alpha_{\epsilon}) + 4 
\frac{(1-(1-\alpha_\epsilon)^n)^2}{n^2 \alpha_\epsilon^2}.
\ee 
\end{theorem}
\begin{proof}
Suppose $\mu f = 0, \mu g = 0$ and $|f|_{\infty}, |g|_{\infty} < 1$. Without 
loss of generality, take $k \ge j$. Then for $X_0 \sim \nu$
\be
\cov(f(X_j),g(X_k)) &= \E[ \E[ f(X_j) g(X_k) \mid \mathcal F_j]] - \E[f(X_j)] 
\E[g(X_k)] \\
&= \E[ f(X_j) \E[ g(X_k) \mid \mathcal F_j]] - \E[f(X_j)] \E[g(X_k)] \\
&\le \E[ f(X_j) (\P^{k-j} g)(X_j)  ] \le 2 (1-\alpha)^{k-j} + 4 
(1-\alpha)^{k+j},
\ee
As before, $\cov(f(X_j)-c_1,g(X_k)-c_2) = \cov(f(X_j),g(X_k))$, so we have for 
general, $\mu$-measurable, bounded functions
\be
\sum_{j=0}^{n-1} \sum_{k=0}^{n-1} \cov(f(X_j),f(X_k)) &\le 2 |f|_*^2 n^2 
S(n,\alpha) + 4 |f|_*^2 \sum_{j=0}^{n-1} \sum_{k=0}^{n-1} (1-\alpha)^{j+k} \\
&\le 2 |f|_*^2 n^2 S(n,\alpha) + 4 |f|_*^2 \frac{(1-(1-\alpha)^n)^2}{\alpha^2} 
\\
\frac{1}{n^2 |f|_*^2} \sum_{j=0}^{n-1} \sum_{k=0}^{n-1} \cov(f(X_j),f(X_k)) &\le 
2 S(n,\alpha) + 4 \frac{(1-(1-\alpha)^n)^2}{n^2 \alpha^2}
\ee
Now, take $\mu_\epsilon f = 0$, then apply the Doeblin condition for 
$\P_\epsilon$, giving
\be
\frac{1}{n^2 |f|_*^2} \sum_{j=0}^{n-1} \sum_{k=0}^{n-1} 
\cov(f(X^\epsilon_j),f(X^\epsilon_k)) &\le 2 S(n,\alpha_\epsilon) + 4 
\frac{(1-(1-\alpha_\epsilon)^n)^2}{n^2 \alpha_\epsilon^2}.
\ee
Applying the triangle inequality and the bias estimate gives the result.
\end{proof}

\subsection{Sharpness}
Now we evaluate the sharpness of the bound in Theorem \ref{thm:ExpStationary} 
by studying a Markov chain and perturbation satisfying the assumptions. Let
\be
 \P = \begin{pmatrix} 1-\beta & \beta \\ \beta & 1-\beta \end{pmatrix}
\ee
for $\beta \le 1/2$. It is easy to verify by direct calculation that the 
invariant measure is $\mu = \begin{pmatrix} \frac12 & \frac12 \end{pmatrix}$ and 
$\P$ satisfies the Doeblin condition with $\alpha = 2\beta$. Any possible 
starting measure $\nu$ can be expressed as $\nu_{\gamma} = (\gamma,1-\gamma)$ 
for some $\gamma \le 1/2$. Then $\|\nu_{\gamma} - \mu\|_{TV} = \frac{1}{2}\left( 
|1/2-\gamma| + |1/2-(1-\gamma)| \right) = \frac{1}{2}-\gamma$ when $\gamma < 
1/2$; note if $\gamma > 1/2$, we get $\gamma - 1/2$. 

Consider the perturbation
\be \label{eq:Peps}
 \P_{\epsilon} = \begin{pmatrix} 1-(\beta-\epsilon) & \beta-\epsilon \\ 
\beta+\epsilon & 1-(\beta+\epsilon) \end{pmatrix},
\ee
which satisfies $\sup_{x \in \mathcal X} \| \P_{\epsilon}(x,\cdot)-  
\P(x,\cdot)\|_{TV} = \epsilon$ and $\sup_{(x,y) \in \X\times \X} \| 
\P_\epsilon(x,\cdot) - \P_\epsilon(y,\cdot) \|_{TV} < 1-2\beta = 1-\alpha$. The 
invariant measure is $\begin{pmatrix} \frac{\beta+\epsilon}{2\beta} & 
\frac{\beta-\epsilon}{2\beta} \end{pmatrix}$, so $\| \mu - \mu_\epsilon \|_{TV} 
= \frac{\epsilon}{\alpha}$. 

A diagonalization of $\P_\epsilon$ is
\be
D_\epsilon &= Q_\epsilon^{-1} \P_\epsilon Q_\epsilon \\
D_\epsilon &= \frac{1}{2a} \begin{pmatrix} \beta+\epsilon & \beta-\epsilon \\ -1 
& 1 \end{pmatrix} \begin{pmatrix} 1 & 0 \\ 0 & 1-2\beta \end{pmatrix}  
\begin{pmatrix} 1 & -(\beta-\epsilon) \\ 1 & \beta+\epsilon \end{pmatrix}
\ee
(see Lawler pg 15-16), so that
\be \label{eq:Pkeps}
\P^k_\epsilon &= \frac{1}{2\beta} \begin{pmatrix} 
(\beta+\epsilon)+(\beta-\epsilon)(1-2\beta)^k & 
(\beta-\epsilon)-(\beta-\epsilon)(1-2\beta)^k \\ (\beta+\epsilon) - 
(\beta+\epsilon) (1-2\beta)^k & (\beta-\epsilon) + 
(\beta+\epsilon)(1-2\beta)^k \end{pmatrix}.
\ee
Therefore
\be
\nu_\gamma \P^k_\epsilon &= \frac{1}{2\beta} \begin{pmatrix} \gamma 
[(\beta+\epsilon)+(\beta-\epsilon)(1-2\beta)^k] + (1-\gamma)[(\beta+\epsilon) - 
(\beta+\epsilon) (1-2\beta)^k] \\ \gamma 
[(\beta-\epsilon)-(\beta-\epsilon)(1-2\beta)^k] + (1-\gamma)[(\beta-\epsilon) 
+ (\beta+\epsilon)(1-2\beta)^k]   \end{pmatrix}' \\
&= \frac{1}{2\beta} \begin{pmatrix} (\epsilon + \beta )+(\beta(2\gamma - 1) - 
\epsilon)(1-2\beta)^k & (\beta-\epsilon) + (\beta(1-2\gamma)+\epsilon) 
(1-2\beta)^k \end{pmatrix}
\ee

The only functions measurable with respect to $\mathcal F_n$ are functions 
taking two values, so without loss of generality, consider a test function $\phi 
= \begin{pmatrix} -1 & 1 \end{pmatrix}$, and take $P_\epsilon$ as in 
\eqref{eq:Peps} then
\be
\left(n^{-1} \sum_{k=0}^{n-1} \phi(X_k) - \mu \phi \right)^2 &=  \left(n^{-1} 
\sum_{k=0}^{n-1} \phi(X_k)\right)^2 \\
&= \left(n^{-1} \sum_{k=0}^{n-1} -\mathbf 1\{X_k=0\} + \mathbf1 
\{X_k=1\}\right)^2 \\
&= \frac{1}{n^2} \sum_{j=0}^{n-1} \sum_{k=0}^{n-1} (\mathbf 1\{X_k=1\} - 
\mathbf 1\{X_k = 0\}) (\mathbf 1\{X_j=1\} - \mathbf 1\{X_j = 0\}) \\
&= \frac{1}{n^2} \sum_{j=0}^{n-1} \sum_{k=0}^{n-1} \mathbf 1\{X_k=X_j\}  - 
\mathbf 1\{X_k \ne X_j\}
\ee
so
\be
\mathbf E \left( n^{-1} \sum_{k=0}^{n-1} \phi(X_k) - \mu \phi \right)^2 = 
\frac{2\|\phi \|_*^2}{n^2} \sum_{k=0}^{n-1} \sum_{j \ge k} \mathbf P [ X_k=X_j ] 
- \mathbf P [X_k \ne X_j].
\ee

Consider $X_j^\epsilon,X_k^\epsilon$ with $j \le k$. Observe that
\be
\mathbf P [X_k^\epsilon = X_j^\epsilon \mid X_j^\epsilon = x] &= \mathbf 
P[X^\epsilon_{k-j} = x \mid X^\epsilon_0 = x]  \\
\mathbf P [X_k^\epsilon \ne X^\epsilon_j \mid X^\epsilon_j = x] &= \mathbf 
P[X^\epsilon_{k-j} \ne x \mid X^\epsilon_0 = x]
\ee
which from \eqref{eq:Pkeps} we can express as
\be
2\beta \mathbf P [X_k^\epsilon = X_j^\epsilon \mid X_j^\epsilon = 0] &= 
(\beta+\epsilon) + (\beta-\epsilon)(1-2\beta)^{k-j} \\
2\beta \mathbf P [X_k^\epsilon = X_j^\epsilon \mid X_j^\epsilon = 1] &= 
(\beta-\epsilon) + (\beta+\epsilon)(1-2\beta)^{k-j} \\
2\beta \mathbf P [X_k^\epsilon \ne X_j^\epsilon \mid X_j^\epsilon = 0] &= 
(\beta-\epsilon)-(\beta-\epsilon)(1-2\beta)^{k-j} \\
2\beta \mathbf P [X_k^\epsilon \ne X_j^\epsilon \mid X_j^\epsilon = 1] &= 
(\beta+\epsilon)-(\beta+\epsilon)(1-2\beta)^{k-j} 
\ee
and therefore
\be
\mathbf P[X_k^\epsilon = X_j^\epsilon] &= 
\frac{1}{(2\beta)^2}((\beta+\epsilon) + 
(\beta-\epsilon)(1-2\beta)^{k-j})((\epsilon + \beta)+(\beta(2\gamma - 1) - 
\epsilon)(1-2\beta)^j) \\
&+ \frac{1}{(2\beta)^2} ((\beta-\epsilon) + 
(\beta+\epsilon)(1-2\beta)^{k-j})((\beta-\epsilon) + 
(\beta(1-2\gamma)+\epsilon) (1-2\beta)^j) 
\ee
and
\be
 \mathbf P[X_k^\epsilon \ne X_j^\epsilon] &= 
\frac{1}{(2\beta)^2}((\beta-\epsilon)-(\beta-\epsilon)(1-2\beta)^{k-j}
)((\epsilon + \beta)+(\beta(2\gamma - 1) - \epsilon)(1-2\beta)^j) \\
&+ \frac{1}{(2\beta)^2} 
((\beta+\epsilon)-(\beta+\epsilon)(1-2\beta)^{k-j})((\beta-\epsilon) + 
(\beta(1-2\gamma)+\epsilon) (1-2\beta)^j).
\ee
Observe that $\beta (2\gamma-1) =  \frac{\alpha}{2} (2\gamma-1) = \alpha \|\mu 
- \nu\|_{TV}$. If we take $\nu_\gamma = \mu_\epsilon$, then $\|\mu-\nu\|_{TV} = 
\frac{\epsilon}{\alpha}$ and $\beta(2\gamma-1) - \epsilon = 0$ so the 
expressions above simplify to
\be
\mathbf P[X_k^\epsilon = X_j^\epsilon] - \mathbf P[X_k^\epsilon \ne 
X_j^\epsilon] = \frac{4 \epsilon^2}{\alpha^2} + (1-\alpha)^{k-j} \left(1- 
\frac{4 \epsilon^2}{\alpha^2} \right)
\ee
so
\be
\mathbf E \left( n^{-1} \sum_{k=0}^{n-1} \phi(X_k) - \mu \phi \right)^2 = 
\frac{4 |\phi|_*^2 \epsilon^2}{\alpha^2} 1  + |\phi|_*^2 S(n,\alpha_\epsilon) 
\left( 1 - \frac{4 \epsilon^2}{\alpha^2} \right)
\ee
since $\alpha_\epsilon = \alpha$. 
Notice that with $\epsilon = 0$, the result is $S(n,\alpha)$, which is half the 
bound in \eqref{eq:ExpectationStationary}. For nonzero $\epsilon$, the bias and 
rate are identical to that in \eqref{eq:ExpectationStationary}, but the constant 
on the second term is smaller: $1-\sfrac{4\epsilon^2}{\alpha^2}$ rather than 2. 
For $\epsilon < \sfrac{\alpha}{4}$, the constant is $\sfrac34$, meaning that 
the constant in \eqref{eq:ExpectationStationary} is too large by at most a 
multiplicative factor of $\sfrac83$. 

\section{Optimal Computation} \label{sec:Optimality}
The (potential) advantage of approximate MCMC is that longer sample paths can 
be obtained in equal computational (wall clock) time. For a transition kernel 
$\P$ satisfying Assumption \ref{ass:Doeblin}, suppose we have a collection of 
approximating kernels $\mathcal A = \{ \P_\epsilon : \epsilon < 
\sfrac{\alpha}{2} \}$, and that the computational cost of taking one step from 
$\P_\epsilon$ is $\tau(\epsilon)$. To avoid uninteresting cases, we will assume 
that $\tau(\epsilon)$ is decreasing in $\epsilon$. Define the \emph{speedup} 
$s(\epsilon)$ of $\P_\epsilon$ by
\be \label{eq:speedup}
s(\epsilon) = \frac{\tau(0)}{\tau(\epsilon)},
\ee
which is the path length one can obtain from $\P_\epsilon$ in the computation 
time it takes to obtain a path of length one from $\P$, or equivalently the 
ratio of a path length from $\P_\epsilon$ to a path length from $\P$ for equal 
computation time. If $\tau(\epsilon)$ is decreasing in $\epsilon$, then 
$s(\epsilon)$ is increasing in $\epsilon$.

Our goal in this section is to define a notion of optimality that leads to 
meaningful generalizations about approximate MCMC and what characteristics of 
$\P_\epsilon$ lead to good performance under the ``nice'' assumptions for which 
we give bounds in the previous section. Although one could conceivably study 
computational optimality under weaker conditions -- for example, geometric 
ergodicity of the original kernel $\mathcal P$ -- it is difficult to make 
meaningful generalizations because two geometrically ergodic Markov chains with 
identical geometric convergence rates can in practice behave very differently 
due to the presence of additional constants in the bounds. As such, we 
restrict ourselves to the uniformly mixing setting. 

The intuition is as follows. Suppose we knew only that $\P$ satisfies 
Assumption \ref{ass:Doeblin}, $\P_\epsilon$ satisfies Assumption 
\ref{ass:UniformTV}, and we are able to evaluate \eqref{eq:speedup}. Then we 
might define the optimal value of $\epsilon$ as the one that achieves the best 
statistical performance in the worst case. We call this the \emph{compminimax} 
$\epsilon$. Specifically, consider a discrepancy function $\DD$ measuring 
statistical performance of a path of length $n$ from $\P$, which we write 
generically as $\DD(\P,n)$. For example, if we take $\D = \D_{TV}$ then
\be
\DD(\P_\epsilon,n) = \left \| \mu - \frac1n \sum_{k=0}^{n-1} \nu \P^k_\epsilon 
\right\|_{TV}.
\ee
For a fixed computational budget $t$, we should compare $\DD(\P_\epsilon,t 
s(\epsilon))$ across different values of $\epsilon$. If we could evaluate $\DD$ 
exactly, then we would choose the value of $\epsilon$ that minimizes 
$\DD(\P_\epsilon,t s(\epsilon))$ and use paths from $\P_\epsilon$ to estimate 
posterior quantities of interest. In general, we can't compute $\DD$ exactly, 
and performing this analysis for a single $\P$ and set of approximations 
$\P_\epsilon$ would not support generalization about features of good 
approximations. Instead, we find $\epsilon$ that minimizes $\DD(\P_\epsilon,t 
s(\epsilon))$ in the worst case scenario, which is equivalent to minimizing a 
sharp upper bound. We call $\epsilon^*$ \emph{compminimax} for computational 
budget $t$ if
\be
\epsilon^* = \inf_\epsilon \sup_{\P_\epsilon} \DD(\P_\epsilon, t s(\epsilon)),
\ee
where the supremum is taken over all $\P_\epsilon$ satisfying Assumption 
\ref{ass:UniformTV} for some $\P$ satisfying \ref{ass:Doeblin}. Because the 
bound in Theorem \ref{thm:approxerr} is sharp, we can use it to evaluate 
$\epsilon^*$ with $\DD = \D_{TV}$. Although the bound in Theorem 
\ref{thm:ExpStationary} may not be sharp, we showed that the bound $\Delta(\mu 
f, \hat f_n^\epsilon) \le S(n,\alpha_\epsilon)$ is achieved for $\epsilon = 0$. 
We therefore perform a ``conservative'' analysis by comparing the bound in 
\eqref{eq:ExpectationStationary} to $S(n,\alpha)$ in determining whether 
$\epsilon^* \ne 0$ for $\DD = \Delta$. This analysis will therefore understate 
the efficiency of $\P_\epsilon$. 

\subsection{Numerical experiments}
Empirical analysis of $\epsilon^*$ requires choices of $\alpha$ and 
$s(\epsilon)$. We consider values between $\alpha=0.1$ and $\alpha=10^{-4}$. 
These values are chosen by considering the upper bound on the $\delta$-mixing 
time $n_{\delta}$ of the chain
\be
n_{\delta} = \inf\{n: \tv{\nu \P^n}{\mu} < \delta\} \label{eq:mixtime}
\ee  
A corollary of Assumption \ref{ass:Doeblin} is that (\ref{eq:mixtime}) is upper 
bounded by $\log(\delta)/\log(1-\alpha)$ when $\| \nu - \mu \|_{TV} = 1$. The 
corresponding worst-case $\delta$-mixing times for a few values of $\delta$ and 
the four values of $\alpha$ considered are given in Table \ref{tab:mixtime}. 
This range of $\alpha$ values gives mixing times between about $45$ and $92,000$ 
for $\delta \in (10^{-2},10^{-4})$, which reflects the empirical performance of 
many MCMC algorithms.\footnote{We acknowledge that the 
criteria used to select burn-in times can result in burn-in periods that do not 
correspond to a mixing time. However, comparing 
mixing times and burn-in periods still provides a useful heuristic, and in most 
cases violation of the criteria used to select a burn-in period is sufficient to 
guarantee that the chain has \emph{not} mixed.} 
% latex table generated in R 3.2.1 by xtable 1.8-0 package
% Sun Jan 24 22:17:00 2016
\begin{table}[ht]
\caption{$\delta$-mixing times for kernels with $d(\mc P) = \alpha$ for different values of $\alpha$ and $\delta$.} 
\label{tab:mixtime}
\centering
\begin{tabular}{l|rrr}
  \hline
 & $\delta=0.01$ & $\delta=0.001$ & $\delta=0.0001$ \\ 
  \hline
$\alpha=0.1$ & 44 & 66 & 87 \\ 
  $\alpha=0.01$ & 458 & 687 & 916 \\ 
  $\alpha=0.001$ & 4,603 & 6,904 & 9,206 \\ 
  $\alpha=0.0001$ & 46,049 & 69,074 & 92,099 \\ 
   \hline
\end{tabular}
\end{table}

We consider four functional forms for $s(\epsilon)$: logarithmic, linear, 
quadratic, and exponential. Constants are chosen such that $s(0) = 1$ and 
$s(\alpha/2) = 100$. Plots of the four functions for $\alpha = 10^{-4}$ are 
shown in Figure \ref{fig:sepsilon}. 

\begin{figure}[h]
\centering
 \includegraphics[width=0.6\textwidth]{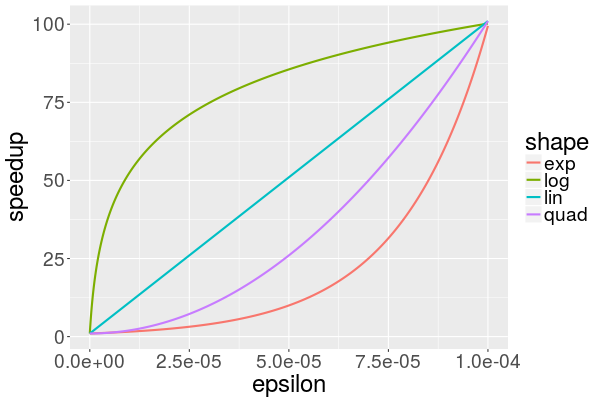}
 \caption{Speedup functions used in analysis of compminimax.} 
\label{fig:sepsilon}
\end{figure}

For each choice of $s(\epsilon)$ and a grid of values of $t \in [1,10^5]$, we 
compute $\epsilon^*$ by minimizing the upper bounds in Theorem 
\ref{thm:approxerr} or Theorem \ref{thm:ExpStationary} with $n = s_{\epsilon} 
t$, depending on the choice of $\D$. When $\D= \Delta$, we put $\epsilon^* = 0$ 
when the minimum of the upper bound in Theorem \ref{thm:ExpStationary} is larger 
than $S(t,\alpha)$. Therefore, the resulting value is exactly $\epsilon^*$ when 
$\D = \D_{TV}$, and a lower bound on $\epsilon^*$ when $\D = \Delta$. When $\D = 
\Delta$, Remark \ref{rem:Aepsilon} indicates that $\alpha_\epsilon \in 
[\alpha-2\epsilon,\alpha+2\epsilon]$, so we perform the experiments assuming 
$\alpha_\epsilon = \alpha$, the center of this interval. 

Results for the numerical experiments are summarized in Figure \ref{fig:opteps}. 
The top two panels show results for $\D_{TV}$. In this case, it is clear that 
over a range of values of $\tau_{\max}$ substantially larger than the mixing 
times, the optimal value of $\epsilon$ is nonzero, regardless of the form of 
$s(\epsilon)$. As $n$ increases, the (approximate) optimal value of $\epsilon$ 
naturally decreases. 

\begin{figure}[h]
 \centering
 \begin{tabular}{cc}
  $D_{TV}, \alpha = 0.1$ & $D_{TV}, \alpha = 10^{-4}$ \\
  \includegraphics[width=0.45\textwidth]{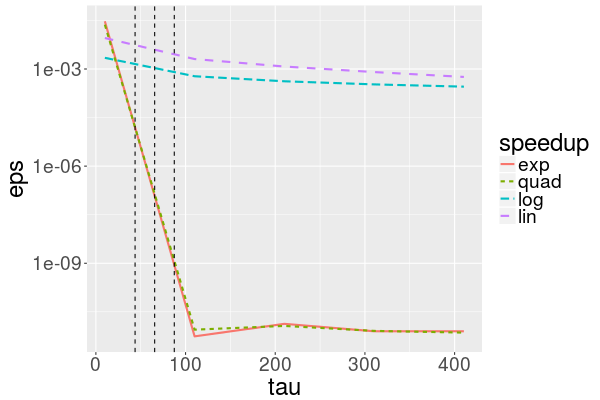} & 
  \includegraphics[width=0.45\textwidth]{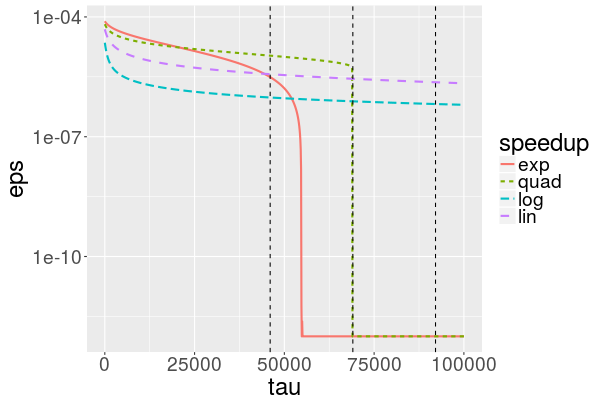} \\
  $D_{L_2}, \alpha = 0.1$ & $D_{L_2}, \alpha = 10^{-4}$ \\
  \includegraphics[width=0.45\textwidth]{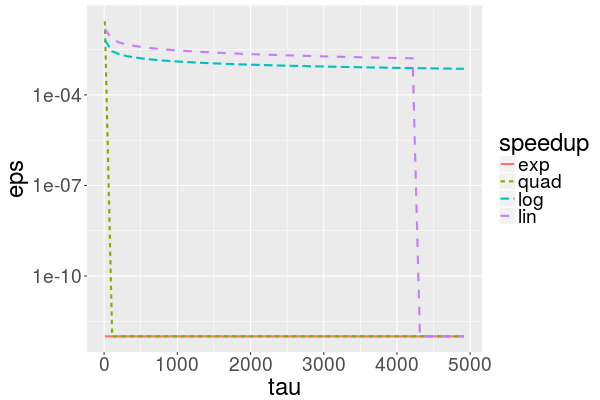} & 
  \includegraphics[width=0.45\textwidth]{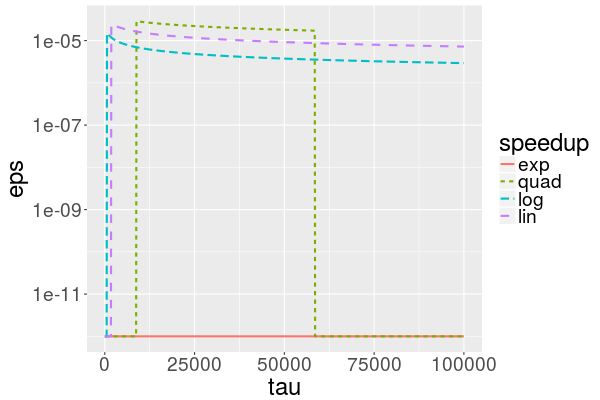} \\
 \end{tabular}
\caption{Plot of $\epsilon_c(\tau_{\max})$ (vertical axis) for values of 
$\tau_{\max} \le 10^5$ (horizontal axis), assuming $\tau_{\mc P}(t) = t$.  
Vertical dashed lines in the top two panels shown at the worst-case 
$\delta$-mixing times for the values of $\delta$ shown in Table 
\ref{tab:mixtime}. Top two panels show results for $D_{TV}$ and bottom two 
panels show results for $D_{L_2}$. Note different horizontal axis scale in the 
left top and bottom panels and that vertical axes use log scale -- the scales 
were chosen to make notable features 
more visible. } \label{fig:opteps}
\end{figure}

The bottom panel in Figure \ref{fig:opteps} shows results for $\D = \Delta$. In 
this case, we assume the chain starts at its stationary distribution, so any 
advantage of $\mathcal P_\epsilon$ in this case is entirely due to the variance 
of the time averages. The choice of $\D = \Delta$ instead of $\D_{TV}$ results 
in (approximate) values of $\epsilon^*$ that are larger at every value of $t$. 
Additionally, values of $\epsilon^*$ are significantly larger than zero well 
beyond the maximum value of $t$ considered in each case ($5,000$ when $\alpha = 
0.1$ and $10^5$ when $\alpha = 10^{-4}$). This reflects the fact that high 
autocorrelations for worst-case functions make variance of MCMC ergodic averages 
the dominating factor in the $\Delta$ error bounds even for relatively long 
sample paths, even when starting from stationarity.

By and large, these experiments were designed to consider a setting in which 
approximate MCMC may be beneficial. We considered cases where a significant 
speedup is available at relatively low cost in terms of the induced 
approximation error, and set $\alpha_\epsilon = \alpha$, which is not the 
``worst case'', since $\alpha_\epsilon$ can be as small as $\alpha-2\epsilon$. 
Under these quite favorable conditions, approximate MCMC may have significant 
advantages. Moreover, it is clear from these experiments that a 
strategy that gradually decreases $\epsilon$ as the chain extends will dominate  
any strategy that uses a fixed $\epsilon$ throughout, assuming that the 
adaptation can be done costlessly. 

\section{Applications} \label{sec:algos}
While it appears that in some circumstances approximate MCMC may offer 
significant benefits, it is unclear whether constructing an accurate 
approximation with sufficient speedup is practicable in real applications. In 
this section, we aim to address this in applications to large sample logistic 
regression and Gaussian process regression. In the first case, $\P_\epsilon$ is 
formed by substituting an approximation to a random matrix based on subsets of 
data, while in the latter a covariance matrix is approximated using an 
approximate partial eigendecomposition. This basic strategy -- approximating an 
expensive matrix computation using a subset of the data -- is broadly 
applicable.

% In the interest of using typical notation from the statistics literature, in this section
% we will use $X$ and $x$ to denote covariate matrices and
% vectors, respectively, and $y$ to denote a response in a regression model. We use
% $\beta, \theta$ to denote parameters, which are also points in the state space
% of the MCMC Markov chain. We hope
% it does not cause any confusion with notation used for Markov processes in the
% previous sections. 

\subsection{Example: logistic regression}
We consider a minibatch approximation of a covariance matrix within a Gibbs 
sampler for logistic regression. We are able to obtain theoretical guarantees on 
approximation error under weaker conditions on the data than 
\cite{korattikara2014austerity}.

We analyze aMCMC based on subsets of data for the logistic regression model 
with likelihood $z_i \sim \Bern{e^{W_i \beta}\{1+e^{W_i \beta}\}^{-1}}$ and 
Gaussian prior $\beta \sim \No{b}{B}$.  
\citet{polson2013bayesian} describe the P\'{o}lya-Gamma distribution on 
$[0,\infty)$, a two-parameter family that we denote $\PG{a_1}{a_2}$. They 
further show that the Gibbs sampler with state variables $x = (\beta,\omega)$ 
and update
\be
 \omega_i \mid \beta &\sim \PG{1}{W_i\beta}, \quad \beta \mid y, \omega \sim 
\No{S_N (X' \kappa + B^{-1} b)}{S_N}, \label{eq:pgsampler}
\ee
where $S_N = (W'\Omega W + B^{-1})^{-1}$, $\kappa = z-1/2$, and $\Omega = 
\mbox{diag}(\omega_1,\ldots,\omega_N)$, has invariant measure with $\beta$ 
marginal equal to $\mu(\beta \mid z,W)$. 

When $N$ is large and $p$ -- the dimension of $\beta$ -- is moderate, the main 
computational bottleneck is calculating $W' \Omega W$. This step has 
computational complexity $\mc O(N p^2)$. An approximating Markov chain that uses 
subsets of size $m$ will reduce the computational complexity of each step to 
$\mc O(m p^2)$, a large computational speedup when $m \ll N$.

We analyze aMCMC with the update rule
\be  \label{eq:scaledlikpg1} 
m &\sim \Binom(N,q) \\
 V \mid \beta &\sim \mbox{Subset}(m, \{1,\ldots,N\}), \\
 \omega_i \mid \beta,V  &\sim \PG{1}{W_i\beta} \quad i \in V, \\ 
 \beta \mid z,\omega,V &\sim \No{S_V W'\kappa }{S_V}, %\label{eq:scaledlikpg3}
\ee
where $S_V = \left( \frac{N}{m} W_V' \Omega_V W_V + B^{-1} \right)^{-1}$ uses 
a subsample-based approximation to $W' \Omega W$ and $\alpha$ may depend on 
$\beta$. \citet{choi2013polya} showed that the algorithm in \eqref{eq:pgsampler} 
satisfies a Doeblin condition.\footnote{The authors only claim the weaker 
uniform ergodicity condition, but they prove a uniform minorization condition on 
the entire state space, which is exactly Assumption \ref{ass:Doeblin}.} Theorem 
\ref{thm:logregerr2} shows that if $m$ is chosen adaptively depending on 
$x$, Assumption \ref{ass:UniformTV} is satisfied.

\begin{theorem}[Error for minibatch approximation] \label{thm:logregerr2}
Define $\lambda^+$ and $\lambda^-$ as the smallest and largest eigenvalues 
of $\frac1N W'\Omega W$. For every $x \in \X$ and $\epsilon > 0$ there 
exists a random constant $K$ and universal constants $C$ and $\sigma^2$ such 
that with probability
\be
2pe^{-\frac{\delta_\epsilon^2/2}{\sigma^2+\left(\frac{1}{Nm}-\frac{1}{N^2} 
\right)C K \delta_\epsilon}}
\ee
we have $\| \P(x,\cdot) - \P_\epsilon(x,\cdot)\|_{TV} < \epsilon$
whenever 
\be
\delta_\epsilon < \frac{(\lambda^- + (N\eta)^{-1})^2}{2p(\lambda^+ + 
\frac{\lambda^-}{2}+(N\eta)^{-1})} \epsilon^2.
\ee
\end{theorem}

The proofs of Theorem \ref{thm:logregerr2} and subsequent results in this 
section are given in the Appendix. The result is established by first 
showing a concentration inequality for the operator norm $\| m^{-1} W_V 
\Omega_V W_V - N^{-1} W'\Omega W\|$, then using this to obtain a total 
variation bound. The constants in Theorem \ref{thm:logregerr2} are made 
explicit in the following remark.
\begin{remark} \label{rem:LogregCons}
With $\theta_i = W_i \beta$, we can take 
\be
K &= \bigvee_i \omega_i, \quad C = \bigvee_i \| W_i \|_2^4 \\
\sigma^2 &= \left( \frac{1}{Nm} - \frac1{N^2} \right) \sum_i \left( \frac{1}{8 
\theta_i^3} \sech^2\left( \frac{\theta_i}2 \right) 
\left( 2 \sinh(\theta_i) + \theta_i (\cosh(\theta_i) - 3) \right) \right) \|W_i 
\|^4_2
\ee
Further, there exist convex functions $V_i(t)$ such that for $t>0$
\be
\Pp[K > t] = 1-\prod_i (1-e^{-V_i(t)}).
\ee
\end{remark}

Suppose that $W_i$ has independent standard Gaussian entries, so that 
$\|W_i\|_2^2$ is a $\chi^2_p$ random variable, with $\| W_i\|_2^4$ roughly 
order $p^2$. In this case, the probability of achieving the approximation error 
is roughly
\be
2pe^{-\frac{\delta_\epsilon^2/2}{\sigma^2+\left(\frac{1}{Nm}-\frac{1}{N^2} 
\right)C K \delta_\epsilon}} &\approx 
2pe^{-\frac{\delta_\epsilon^2/2}{\left(\frac{1}{Nm}-\frac{1}{N^2} 
\right) N p^2 K \delta_\epsilon}} \\
&\approx 2pe^{-\frac{\delta_\epsilon}{2 \left(\frac{1}{Nm}-\frac{1}{N^2} 
\right) N p^2 K}} 
\ee
so if we take $m = N^{\gamma}$ for $\gamma < 1$ we obtain 
\be
\Pp[ \| \P(x,\cdot) - \P_\epsilon(x,\cdot)\|_{TV} < \epsilon ] \approx 
2pe^{-\frac{H \epsilon^2 N^\gamma}{2 p^2 K}} 
\ee
for 
\be
H = \frac{(\lambda^- + (N\eta)^{-1})^2}{2p(\lambda^+ + 
\frac{\lambda^-}{2}+(N\eta)^{-1})},
\ee
so that rates of $e^{-N^\gamma}$ can be achieved with minibatches of size 
$N^\gamma$. This also implies that $\epsilon$ must go to zero no faster than 
$N^{-\gamma/2}$ to achieve the error condition in large samples using 
minibatches of size $N^\gamma$.  

In applications, the constants will of course be important, especially when 
$\gamma$ is small. The hardest constant to estimate is $K$, so we perform a 
simulation study of its magnitude. We sample $N = 10^6$ vectors $W_i$ of length 
$p=10$ with iid standard normal entries, then repeatedly sample 
\be
\beta \sim \No{0}{I_p}, \quad \omega_i \sim \PG{1}{W_i \beta}
\ee
and compute $K=\max_i \omega_i$, $\lambda^+$, and $\lambda^-$. Results for 
1,000 replicate samples are shown in Figure \ref{fig:ConstantSims}. The 
distribution of $K$ is concentrated between 2 and 4, $\lambda^-$ between 0.05 
and 0.2, and $\lambda^+$ between 0.14 and 0.24. Under these admittedly 
idealized conditions, the constants in the bounds are not particularly onerous.

\begin{figure}[h]
 \centering
 \includegraphics[width=\textwidth]{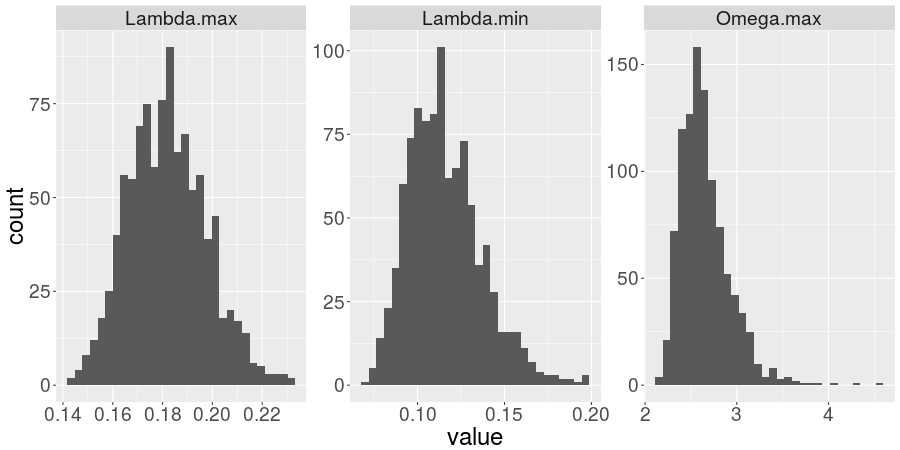}
 \caption{Histograms of $\lambda^+$ (left), $\lambda^-$ (center), and $K$ 
(right) for simulation study described in text.} \label{fig:ConstantSims}
\end{figure}

\subsection{Computational example: record linkage dataset}
We applied the sampler in \eqref{eq:scaledlikpg1} to a record linkage dataset
from the Epidemiological Cancer Registry of North Rhine-Westphalia 
\cite{sariyar2011controlling}.\footnote{In the applications we use a fixed 
value of $m$ rather than sampling it from a Binomial.} We use a portion of the 
data consisting of 2.2 
million observations.\footnote{Code for the applications is available at 
https://github.com/jamesjohndrow/amcmc-jmmd} The
covariates contain information about the similarity of each record pair, and the 
outcome is a binary indicator of whether the two records pertain to the same 
individual. To produce the final covariate data used in this example, we 
generate all two-way interactions and quadratic terms, then select the first 
twenty principal components of this data matrix. This alleviates 
high collinearity in the pairwise interaction matrix.  

The algorithm in \eqref{eq:scaledlikpg1} is applied for a range of minibatch sizes
between 11,000 and 2.2 million (the full data). For each minibatch size, we run the
algorithm for 5,000 iterations. To assess peformance of aMCMC, the final 4,000 
samples from the exact MCMC are treated as samples from the target measure.
We then estimate discrepancy measures for the path obtained from aMCMC with
each minibatch size as a function of computation (wall clock) time. The two discrepancy
measures used are the two sample Anderson-Darling statistic and the 
squared Frobenius norm
of the difference between the path covariance matrix for $\beta$ and the covariance
matrix based on the full sample of 4,000 from the exact sampler. Specifically, define
\be \label{eq:EmpCov}
\hat \beta_n &= \frac1n \sum_{k=0}^{n-1} \beta_k, \quad 
\hat \Sigma_n = \frac1n \sum_{k=0}^{n-1} (\beta-\hat \beta_k) (\beta-\hat 
\beta_k)'
\ee
then we compute $\| \hat \Sigma^{(\epsilon)}_n - \hat \Sigma_{4000} \|_F^2$ as 
a discrepancy measure, where $\hat \Sigma^{(\epsilon)}_n$ is obtained by 
computing \eqref{eq:EmpCov} for the minibatch algorithm for different values of
$n$ and minibatch sizes.

Results are shown in Figure \ref{fig:DLogReg}. Conclusions are somewhat different
for the two discrepancy measures. The minibatching algorithm performs better 
for estimation of the covariance matrix, with the sample sizes of 440,000 or
smaller being optimal for computation times up to 5,000 seconds. The exact
sampler does not become optimal with respect to this discrepancy measure until
computation time exceeds 9,000 seconds. On the other hand, the exact MCMC 
is optimal with respect to the Anderson-Darling statistic after less than 1,000 seconds
of computation time. 

\begin{figure}[h]
\centering
\begin{tabular}{cc}
$\| \hat \Sigma_{4000} - \hat \Sigma^{(\epsilon)}_n \|_F^2 $ & Anderson-Darling 
statistic \\
\includegraphics[width=0.5\textwidth]{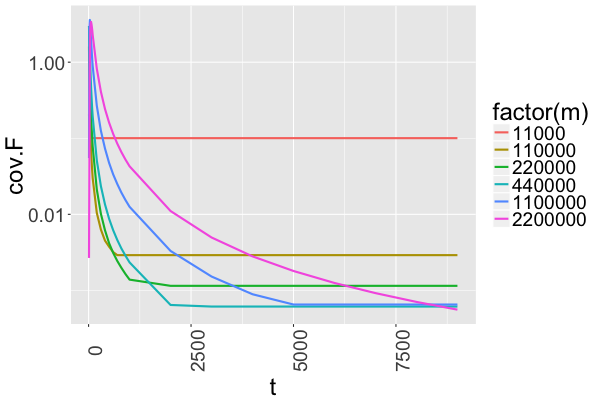} & \includegraphics[width=0.5\textwidth]{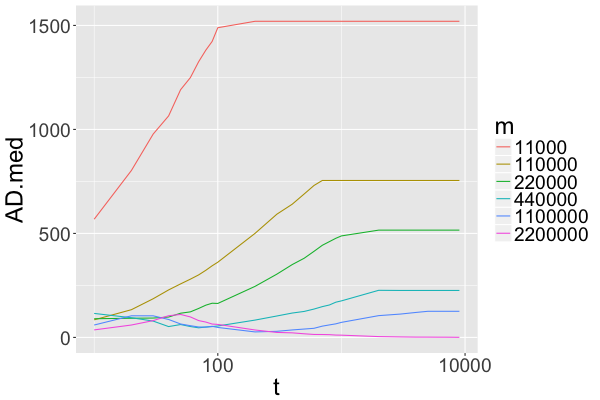} 
\end{tabular}
\caption{Measures of discrepancy based on the covariance matrix (left) and the two sample 
Anderson-Darling statistic (right) as a function of computation time (horizontal axes, left panel)
or log computation time (horizontal axis, right panel). The
plot at right shows the median statistic over the twenty components of $\beta$.} \label{fig:DLogReg}
\end{figure}

This suggests that the performance and optimality of aMCMC
varies considerably depending on the discrepancy measure used. This was 
already evident in the discussion of compminimax optimality. When started
from the invariant measure, $\epsilon=0$ is always optimal with respect to
$\D_{TV}$, but $\epsilon > 0$ will generally be optimal for small enough 
computational budget regardless of the starting measure when using $\Delta$ 
as the discrepancy measure. Here, we see that aMCMC performs better
when the goal is estimation of the second moment of the identity function
than when the goal is estimation of a large class of $\mu$-measurable 
functions, for which small error would correspond to small values of
the Anderson-Darling statistic.

\subsection{Low-rank approximations to Gaussian processes} \label{sec:gp}
Exact MCMC algorithms involving Gaussian processes on samples of size $N$ 
scale as $\mc O(N^3)$, 
leading to numerous proposals for approximations. Prominent examples include the 
predictive process \cite{banerjee2008gaussian} and subset of regressors 
\cite{smola2001sparse}, which both employ low-rank approximations to the 
Gaussian process covariance matrix. 

\subsubsection{Model}
Consider the nonparametric regression model
\be
 z_i = f(W_i) + \eta_i, \qquad \eta \sim \No{0}{\sigma^2 I_n}, \qquad i = 
1,\ldots,N, \label{eq:npreg}
\ee
where $z_i$ are responses, $W_i$ are $p \times 1$ covariate vectors, 
and $f$ is an unknown function. A typical Bayesian approach  assigns a Gaussian 
process prior to $f$, $f \sim \gp(\mu(\beta),c(\gamma))$, with 
$\mu(\cdot;\beta)$ a mean function with parameter $\beta$ and 
$c(\cdot,\cdot;\gamma)$ a covariance function parametrized by $\gamma$, so that 
for $W_1,W_2$, $\mathbf E f(W_1) = \mu(W_1;\beta)$ and 
$\cov(f(W_1),f(W_2)) = c(W_1,W_2;\gamma)$. Here we will assume $\mu(W;\beta) 
\equiv 
0$, so that the model parameters consist of $x = 
(\sigma^2,\gamma)$. Although we focus on model (\ref{eq:npreg}), our analysis 
applies to general settings involving Gaussian processes (e.g., for spatial 
data).

The covariance kernel $c(W_1,W_2;\gamma)$ is positive definite, so that the $N 
\times N$ covariance matrix $S$ given by $S_{ij} = c(W_i,W_j;\gamma)$ is full 
rank. However, as noted by \citet{banerjee2012efficient}, in many cases when 
$N$ 
is large, the matrix $S$ is poorly conditioned and nearly low-rank. This 
motivates low-rank approximations of $S$.  As an example, consider the squared 
exponential kernel $c(W_1,W_2;\gamma) = \tau^2 e^{-\phi \|W_1-W_2\|^2}$, 
with $\gamma = (\tau^2,\phi)$ consisting of a decay parameter $\phi$ and scale 
$\tau^2$. In this case we write $S = \tau^2 \Sigma$, where $\Sigma_{ij} = 
e^{-\phi \|W_i-W_j\|^2}$, and we have $x = (\sigma^2,\tau^2,\phi)$. We 
adopt the common prior structure
\be \label{eq:gpprior}
 \phi &\sim \mbox{DiscUnif}(\phi_1,\ldots,\phi_d) \\
 \tau^{-2} &\sim \Gam{a_{\tau}}{b_{\tau}}, \\  
\sigma^{-2} &\sim \Gam{a_{\sigma}}{b_{\sigma}}.
\ee

We consider a MCMC sampler (e.g. \cite{finley2009improving}) which 
iterates
\begin{enumerate}
 \item Sample $\sigma^2, \tau^2 | z, \phi$ using a Metropolis-Hastings step.
 We use a Gaussian random walk on $(\log(\sigma^2),\log(\tau^2))$ as proposal.
 \item Set $p_l ~\propto~ |\tau^2 \Sigma^{(l)}+\sigma^2 I_n|^{-1/2} 
e^{-z' (\tau^2 \Sigma^{(l)}+\sigma^2 I_n)^{-1} z/2 }$, where 
$\Sigma^{(l)}_{ij} =e^{-\phi_l \|W_i-W_j\|^2}$, and sample 
 $$\phi \sim \Disc{\{\phi_1,\ldots,\phi_d\}}{(p_1,\ldots,p_d)}.$$
\end{enumerate}

\subsubsection{Approximate MCMC for Gaussian processes}
We replace $\Sigma$ with a low-rank approximation $\tSigma \approx \Sigma$ to 
construct a transition kernel $\P_\epsilon(x,\cdot)$.  We focus on 
approximations of the form 
\be
 \Sigma \approx \Ueps \Lameps \Ueps' = \tSigma,  \label{eq:halko}
\ee
where $\Ueps$ is orthonormal, and $\Lameps$ is nonnegative and diagonal. All 
of the steps of the MCMC sampler contain the quadratic form $z' (\tau^2 
\Sigma + \sigma^2 I)^{-1} z$.
The approximation instead uses $\Psieps = (\tau^2\tSigma + \sigma^2 I)^{-1}.$

For algorithms in this class, we obtain the result in Theorem \ref{thm:gperror}.
\begin{theorem}[Gaussian process approximation error bounds] \label{thm:gperror}
If $\tSigma$ is a partial rank-$k$ eigendecomposition of $\Sigma$, then for 
every $\epsilon > 0$, there exists a $\P_\epsilon(x,\cdot)$ that replaces 
$\Sigma$ with 
$\tSigma$ achieving $\|\Sigma-\tSigma\|<\delta$, where 
$\delta$ depends on $x$, such that 
\be
\sup_{x \in \X} \tv{\P(x,\cdot)}{\P_\epsilon(x,\cdot)} < \epsilon. 
\label{eq:tvgpkern}
\ee
Further, we have $\epsilon = \bigO{\delta}$.
\end{theorem}

In practice, when $N$ is so large that one cannot compute an exact partial 
eigendecomposition,  Algorithm 2 of \citet{banerjee2012efficient} provides an 
accurate approximation, so the approximation is feasible regardless of $N$.
The algorithm of \cite{banerjee2012efficient} is equivalent to the 
\emph{adaptive randomized range finder} (Algorithm 
4.2) combined with the \emph{eigenvalue decomposition via Nystr\"{o}m 
approximation} (algorithm 5.5) in \citet{halko2011finding}.  Algorithm 2 attains 
approximation error $\|\Sigma - \tSigma\|_F < \delta$ with probability 
$1-10^{-d}$ where both $\delta$ and $d$ can be specified. Not all low-rank 
approximations of $\Sigma$ approximate a partial eigendecomposition, so Theorem 
\ref{thm:gperror} suggests a possible advantage of Algorithm 2 of 
\cite{banerjee2012efficient} over alternatives.The following remark describes 
the achievable rates in $\delta$ as a function of $\epsilon$ and $N$.
\begin{remark}[Rates for aMCMC for Gaussian process] \label{rem:gprates}
Define $k \equiv \rank(\Sigma_\epsilon)$. Controlling $\delta$ to satisfy 
Assumption 
\ref{ass:UniformTV} requires that
 \be
\bigg| &e^{-\frac{N-k}{2} \left[ 
\frac{\sigma^2_x-\sigma^2_y}{\sigma_y^2\sigma^2_x}  - 
\log\frac{\sigma_y^2}{\sigma^2_x} \right]} -e^{-\frac{N-k}{2} \left[ 
\frac{\tau^2_x \delta + \sigma^2_x - \tau_y^2 
\delta - \sigma_y^2}{(\tau_y^2 \delta + \sigma_y^2)(\tau^2_x \delta + 
\sigma^2_x)} -  \log\frac{\tau_y^2 \delta + \sigma_y^2}{\tau^2_x \delta + 
\sigma^2_x} \right]}  \bigg|
 \ee
be small, where $\sigma_y^2, \tau_y^2$ are the proposed values and  
$\sigma^2_x,\tau^2_x$ the current state in the Metropolis-Hastings algorithm.  
To achieve constant approximation error, $\delta$ must decrease with $N$; if 
the spectrum of $\Sigma_{\epsilon}$ decays rapidly, the decrease can be slow. 
In addition, a smaller value of $\delta$ is required when $\tau^2$ is large 
relative to $\sigma^2$, suggesting that a higher signal to noise ratio requires 
better approximations.
\end{remark}
 
One step of $\P_\epsilon$ scales as $\bigO{N^2 k}$, so increasing $k$ to 
achieve a better approximation has computational cost $N^2$. However, the 
relationship between $k$ and $\delta$ -- and therefore between $k$ and 
$\epsilon$ -- depends on the spectrum of $\Sigma$. The speedup is proportional 
to the rate at which the eigenvalues of $\Sigma$ decay, so virtually any rate 
is conceivable, and will depend on the points at which the process is sampled 
and the values of $\phi$ that account for most of the posterior probability.

\subsection{Computational example: Abalone dataset}
We apply the aMCMC algorithm for Gaussian process regression that utilizes
low-rank approximations of the form \eqref{eq:halko} to the Abalone dataset
\cite{nash1994population, waugh1995extending}. The data consist of $N=4,177$ 
observations of the age $y$ of Abalone,
with $p=8$ covariates used as predictors. Following common practice in the 
spatial statistics literature, we choose a discrete uniform prior on $\phi$ with 
mass on twenty points in the range $-\log(0.01)/(\zeta d_{\max})$, where
$d_{\max}$ is the maximum Euclidean distance between sample points and
$\zeta$ ranges from $0.1$ to $0.9$. This results in the correlation between
observations decaying to 0.01 at distances between 0.1 times $d_{\max}$ and 
$0.9$ times $d_{\max}$. This results in twenty possible values of $\Sigma$, 
each of which has a distinct eigenstructure. The accuracy of a rank $k$
approximation to $\Sigma$ in the Frobenius norm will depend on the
ratio
\be \label{eq:PerExpl}
\frac{\sum_{i=1}^k \lambda_i}{\sum_{i=1}^N \lambda_i},
\ee
the proportion of variation explained by the first $k$ eigenvectors of
$\Sigma$. Figure \ref{fig:eigenvalues} shows \eqref{eq:PerExpl} for
a selection of values of $\zeta$ used in the analysis.

\begin{figure}[h]
\centering
\includegraphics[width=0.75\textwidth]{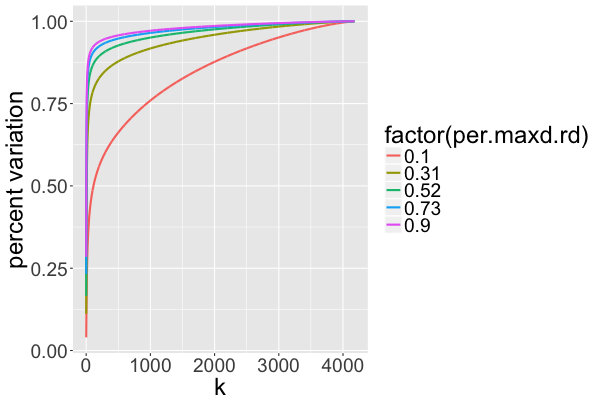}
\caption{Proportion of variation explained (vertical axis) as a 
function of the rank $k$ of the approximation (horizontal axis). 
Different values of $\zeta$ (the percent of maximum distance
at which the correlation decays to 0.01) are indicated
by color. } \label{fig:eigenvalues}
\end{figure}

Computation was performed for the low-rank Gaussian process approximations as 
described in section \ref{sec:gp}. We use an adaptive Metropolis strategy to 
tune the proposal covariance for $(\sigma^2,\tau^2)$.  Here, instead of 
approximation to the posterior, we evaluate the performance of the approximation 
by how close $\P$ is to $\P_\epsilon$ in the total variation norm. Since $\P$ 
and $\P_\epsilon$ use the same proposal kernel
$Q$, we have 
\be
\| \P - \P_\epsilon\|_{TV} &= \sup_{\varphi: |\varphi|<1} \int  
(\varphi(x)-\varphi(y)) 
|\alpha(x,y)-\alpha_\epsilon(x,y) | Q(x,y) 
dy \\
&\le 2 \int |\alpha(x,y)-\alpha_\epsilon(x,y) | 
Q(x,y) dy,
\ee
where $\alpha(x,y)$ is the acceptance ratio for moving from
$x$ to $y$ for the exact algorithm, and 
$\alpha_\epsilon(x,y)$
is the analogue for aMCMC. So to upper bound $\| \P - \P_\epsilon\|_{TV}$, it is 
enough to estimate the expected value of the absolute difference between the 
acceptance ratios at any point, with the expectation taken with respect to the 
proposal. Obviously, this is infeasible to do at every $x$. Instead, we compute 
$|\alpha(x,y)-\alpha_\epsilon(x,y) |$ for paths $x_0,x_1,\ldots,x_n$ taken from 
$\P$, an approximation of the expectation of this quantity with respect to the 
invariant measure. A similar analysis applies to the absolute differences 
between the transition probabilities for the griddy Gibbs update for $\phi$, 
which we also compute.

Figure \ref{fig:GPDists} shows results for four different values of $k$ for both the
absolute difference between the Metropolis acceptance ratios and the maximum
absolute difference between the griddy Gibbs transition probabilities. Clearly, 
as $k$
increases, the distribution of the error in approximation of both quantities  
shifts toward zero. This shift is expected as the approximation becomes more 
accurate. However, even when using $k=2,000$ of the 4,177 components, the 
Metropolis acceptance ratios can still be arbitrarily bad with non-negligible 
probability. This likely reflects the fact that noticeable posterior mass is 
placed on the smallest values of $\phi$, which correspond to the largest values 
of $\zeta$ in Figure \ref{fig:eigenvalues}. For such values of $\phi$, the 
spectrum of $\Sigma$ decays slowly, and $k \approx N$ is necessary to make an 
accurate approximation. It also hints at a pattern we have frequently observed 
in working with aMCMC algorithms: likelihood ratios tend to be unforgiving of 
small perturbations. Nonetheless, if one is willing to restrict the analysis to 
cases where dependence in the regression function is more local, this analysis 
suggests that accurate low-rank approximation is possible.

\begin{figure}[h]
\centering
\includegraphics[width=0.9\textwidth]{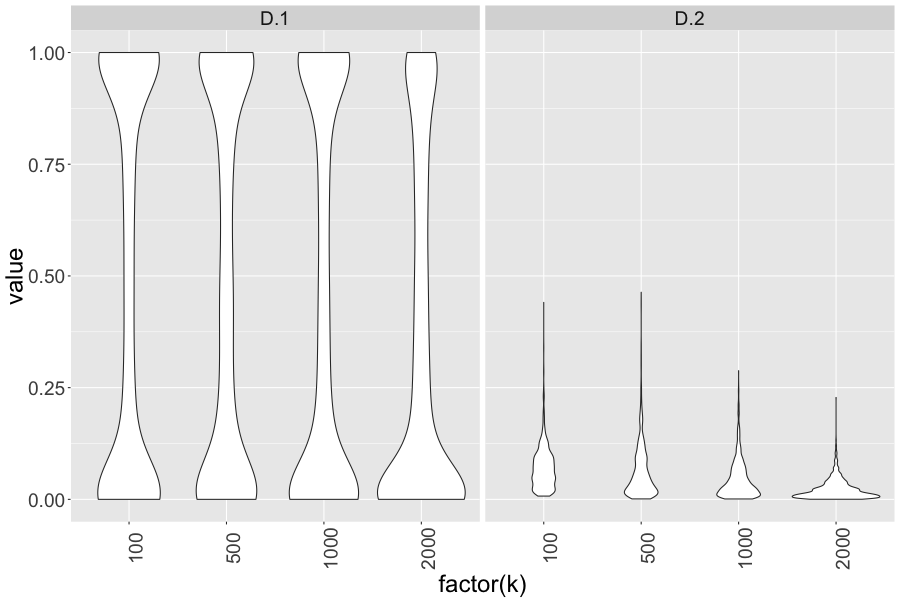}
\caption{Distribution of 
$|\alpha(x,y)-\alpha_\epsilon(x,y)|$ (left panel) and 
of the maximum absolute difference in griddy Gibbs transition probabilities 
(right panel) for different numbers of components used in the low rank 
approximation to $\Sigma$ (k). } \label{fig:GPDists}
\end{figure}

\section{Discussion}
A large literature exists showing general approximation error bounds for Markov 
chain Monte Carlo, including more recent work on perturbation bounds applied to 
approximate MCMC. In general, these bounds provide guarantees about long-time 
dynamics of the approximation, but do not give much guidance about whether such 
approximations are actually advantageous with limited computational budgets. 
Working with simple, sharp bounds under ideal conditions, we propose general 
principles regarding when approximate MCMC is likely to be useful. In 
particular, when the exact algorithm mixes slowly -- i.e. when $\alpha$ in Assumption
\ref{ass:Doeblin} is close to 0 -- and significant speedups are available that introduce 
small approximation error, approximations are likely to 
enjoy an advantage over exact MCMC. It is important to note that ergodicity 
conditions such as uniform or geometric ergodicity do not guarantee ``fast
mixing,'' since the convergence rate can be arbitrarily close to one -- and in many
applied settings MCMC apparently mixes very slowly.

Our analysis of minibatching suggests that 
it does not provide a general ``formula'' for constructing approximations with 
significant computational advantages, though in certain cases it may be useful.
This is consistent with the analysis of \citet{bardenet2014towards} and 
\citet{pillai2014ergodicity}. Non-random subsampling strategies and 
bias-correction can improve performance \cite{quiroz2014speeding} relative to 
random sampling, and when the resulting likelihood estimators are 
unbiased there is an obvious connection with the literature on 
pseudo-marginal MCMC that merits further investigation. 

Constructing very efficient and scalable MCMC algorithms for complex target 
distributions will likely necessitate a combination of strategies including, but 
not limited to, approximate MCMC. We hope the insights offered here will provide 
some guidance to practitioners in determining when approximations could be 
useful and how to go about constructing them such that basic properties of the 
time averages are preserved.

\section*{Acknowledgments}
The authors thank Galen Reeves for suggesting the concentration argument used 
in proving Theorem \ref{thm:logregerr2}.

\begin{appendix}
\section{Proof of Theorem \ref{thm:logregerr2}}

We will show $\sup_{x \in \X} \| \P(x,\cdot) - \P_\epsilon(x,\cdot) \|_{TV} 
< \epsilon$ with high probability. Here, $\P$ is the transition kernel based on 
the full sample of $N$ observations for the Gibbs sampler in 
\eqref{eq:pgsampler}, and $\P_\epsilon$ uses subsets of data of size $m \le N$ 
to approximate $W'\Omega W$ by $\frac{N}{m} W_V' \Omega_V W_V$, in accordance 
with the update rule in \eqref{eq:scaledlikpg1}.  

Define
\be
\Sigma_N(\beta) &= \frac{1}{N} W' \Omega W, \qquad \Sigma_V(\beta) = 
\frac{1}{|V|} W_V' \Omega_V W_V \\
S_N(\beta) &= \frac{1}{N} (\Sigma_N(\beta) + B^{-1}/N)^{-1}, \qquad S_V(\beta) 
= \frac{1}{N}(\Sigma_V(\beta) + B^{-1}/N)^{-1};
\ee
we will sometimes suppress dependence on $\beta$ for notational convenience. 
Recall that the distribution of $\beta$ given $\omega$ in the Gibbs update is 
$\No{S_N 
X'\kappa}{S_N}$, with $\kappa=z-1/2$. Let $\mc N(\cdot;\mu,\Sigma)$ be the 
measure 
induced by a normal random variable with mean $\mu$ and covariance $\Sigma$. 

We first show that for every $\delta$ there exists $m$ 
for which
\be
\norm{\Sigma_N - \Sigma_V} \le \delta
\ee
with high probability, where for matrices $\|\cdot\|$ denotes the spectral 
norm. We then apply this to bound the 
Kullback-Leibler divergence 
\be
\KL{\mc N_V}{\mc N_N} &= 
\frac{1}{2} \bigg(\tr{ S_N^{-1} S_V} - p + \Log{\frac{|S_N|}{|S_V|}} + Q 
\bigg),
\ee
with $\mc N_V = \mc N(\cdot;S_V W'\kappa,S_V)$, $\mc N_N = \mc N(\cdot;S_N 
W'\kappa,S_N)$, and 
\be
Q = (S_N W'\kappa - S_V W' \kappa)' S_N^{-1} (S_N W'\kappa - S_V W'\kappa),
\ee
and use Pinsker's inequality to obtain a total variation bound. 
We will choose $\delta$ as a function of $\epsilon$ and $x$ 
to obtain $\| \P(x,\cdot) - \P_\epsilon(x,\cdot) \| < \epsilon$; thus, the 
supremum is controlled by adaptive choice of $\delta$. When this requires 
$m>N$, put $V = \{1,\ldots,N\}$ and obtain the exact kernel. 

We proceed in four steps:
\begin{enumerate}
 \item Showing we can control $\|\Sigma_V - \Sigma_N\|$ with high probability;
 \item Obtaining bounds on the eigenvalues of $\Sigma_V$ and $\Sigma_N$ when 
$\|\Sigma_V - \Sigma_N\|<\delta $;
 \item Using (a) and (b) to control the KL; and
 \item Showing how to choose $\delta$ as a function of $\beta$ to achieve 
uniform control of $\|\P(x,\cdot) - \P_\epsilon(x,\cdot) \|_{TV}$.
\end{enumerate}

\subsection*{Part (a): Control of $\| \Sigma_V - \Sigma_N \|$}
Suppose we take a minibatch of \emph{random} size $m$, with $m \sim 
\Binom(N,q)$. 
We can represent the matrix $W' \Omega W$ as 
\be
W' \Omega W = \sum_{i=1}^n \omega_i W_i W_i' \equiv \sum_{i=1}^n A_i
\ee
Define
\be
M_N \equiv \frac1m \sum_{i=1}^N \xi_i A_i - \frac1N \sum_i A_i = \sum_i 
\left(\frac{\xi_i}{m} - \frac1n \right) A_i 
\ee
and observe
\be
\E \left(\frac{\xi_i}{m} - \frac1N \right) A_i &=  \left(\frac{\E \xi_i}{m} - 
\frac1N \right) \E A_i \\
&=  \left(\frac{q}{m} - \frac1N \right) \E A_i = \left(\frac{m/N}{m} - 
\frac1N \right) \E A_i = 0.
\ee
Define
\be
\sigma^2 &\equiv \left\| \sum_{i=1}^N \E \left(\frac{\xi_i}{m} - \frac1N 
\right)^2 A_i^2 \right\| =  \left\| \sum_{i=1}^N \left( \frac{q}{m^2} -2 
\frac{q}{mN} + \frac1{N^2} \right)  \E A_i^2 \right\| \\
&=  \left\| \sum_{i=1}^N \left( \frac{1}{Nm}  - \frac1{N^2} \right)  \E A_i^2 
\right\| = \left( \frac{1}{Nm}  - \frac1{N^2} \right) \left\| \sum_{i=1}^N  \E 
A_i^2 \right\| 
\ee
and then observe that with $W_i \beta = \theta_i$
\be
\| A_i \| &= \| \omega_i W_i \|_2^2 = \omega_i^2 \|W_i \|_2^2 \\
 \left\| \sum_i \E A_i^2  \right\| &= \left\| \sum_i \E \omega_i^2 \E 
(W_iW_i')^2 \right\| \\
 &\le \sum_i \E \omega_i^2 \|W_i W_i'\|^2 = \sum_i \E \omega_i^2 \|W_i \|^4_2 
\\
 &= \sum_i \left( \frac{1}{4 \theta_i^3} (\sinh(\theta_i) - \theta_i) 
\sech^2\left( \frac{\theta_i}{2} \right) + \frac{1}{4 (\theta_i)^2} 
\tanh^2\left( \frac{\theta_i}{2} \right)  \right) \|W_i \|^4_2, \\
&= \sum_i \left( \frac{1}{8 \theta_i^3} \sech^2\left( \frac{\theta_i}2 \right) 
\left( 2 \sinh(\theta_i) + \theta_i (\cosh(\theta_i) - 3) \right) \right) \|W_i 
\|^4_2, \\
\ee
giving an explicit expression for $\sigma^2$. The expectation of $\omega_i^2$ 
is available from the Laplace transform, see \cite{polson2013bayesian}. 
Putting 
\be \label{eq:ConcCons}
K &= \max_i \omega_i , \quad C = \max_i \|W_i \|^4_2  \\
L &= \left( \frac{1}{Nm}  - \frac1{N^2} \right) C K
\ee
we obtain from the matrix Bernstein inequality \cite[Theorem 1.4]{tropp2012user}
\be \label{eq:Conc}
\Pp [ \lambda_{\max}(M_N) \ge t  ] \le 2 p e^{-\frac{t^2/2}{\sigma^2 + Lt/3}}.
\ee
Using the relation $\|M \|^2_F \le p \lambda_{\max}(M)^2$, we obtain a high 
probability bound on $\|M\|^2_F$.

\subsection*{Part (b) : Control of eigenvalues of $\Sigma_N$ and $\Sigma_V$}
If $\|\Sigma_V - \Sigma_N\| < \delta$ then 
\be
\|\Sigma_V + B^{-1}/N - (\Sigma_N + B^{-1}/N) \| = \|\Sigma_V-\Sigma_N\| 
\le \delta.
\ee
Now, use $\|\Sigma_V - \Sigma_N\|_F^2 \le p \|\Sigma_V-\Sigma_N\|^2$, and 
that 
$\Sigma_V, \Sigma_N$ are Hermitian, and apply the Hoffman-Weilandt inequality 
\cite{bhatia2013matrix,hoffman1953variation,tao2015eigenvalues}, which 
ensures existence of a permutation $\rho$ of 
the eigenvalues of $\Sigma_V$ such that
\be
 \sum_{j=1}^p (\lambda_{\rho(j)}(\Sigma_V) - \lambda_j(\Sigma_N))^2 < 
\|\Sigma_V 
- \Sigma_N\|_F^2 \le (\sqrt{p} \|\Sigma_V - \Sigma_N\|)^2 \le p \delta^2,
\ee
where $\lambda_j(\Sigma)$ is the $j$th eigenvalue of the matrix $\Sigma$. So 
there exists a $j$ such that 
\be \label{eq:maxeigdist}
(\lambda_{\max}(\Sigma_V) - \lambda_j(\Sigma_N))^2 <  p \delta^2,
\ee
where $\lambda_{\max}(\Sigma_V)$ is the largest eigenvalue of $\Sigma_V$. This 
implies that 
\be \label{eq:maxeigub}
\lambda_{\max}(\Sigma_V) < \lambda_{\max} (\Sigma_N) + \sqrt{p} \delta. 
\ee
This is immediate if $j=1$ in (\ref{eq:maxeigdist}). If $j>1$ in 
(\ref{eq:maxeigdist}), then we must have (\ref{eq:maxeigub}), since otherwise 
\be
(\lambda_{\max}(\Sigma_V) - \lambda_j(\Sigma_N))^2 \ge  
(\lambda_{\max}(\Sigma_N) + 
\sqrt{p} \delta - \lambda_j(\Sigma_N))^2 \ge p \delta^2.
\ee
Furthermore, there exists a $k$ for which
\be
(\lambda_{\min}(\Sigma_V) - \lambda_k(\Sigma_N))^2 <  p \delta^2, 
\ee
with $\lambda_{\min}(\Sigma_V)$ the smallest eigenvalue of $\Sigma_V$, implying
\be
\lambda_{\min}(\Sigma_V) > \lambda_{\min}(\Sigma_N) - \sqrt{p} \delta
\ee
by analogous argument. So if 
\be
\delta < \frac{1}{\sqrt{p}} \frac{\lambda_{\min}(\Sigma_N)}{2}, 
\ee
we have 
\be
\lambda_{\min}(\Sigma_V) > \frac{\lambda_{\min}(\Sigma_N)}{2} \equiv 
\ell_{\min}(\beta)
\ee
ensuring the smallest eigenvalue of $\Sigma_N$ is bounded away from zero, and 
\be
\lambda_{\max}(\Sigma_V) < \lambda_{\max}(\Sigma_N) + 
\frac{\lambda_{\min}(\Sigma_N)}{2} \equiv \ell_{\max}(\beta)
\ee
So with $B = \eta I_p$,
\be
\lambda_{\min,N} \equiv \lambda_{\min}(N S_N(\beta)) 
&\ge \frac{1}{\ell_{\max}(\beta) + (N\eta)^{-1}}, \\
\lambda_{\max,N} \equiv \lambda_{\max}(N S_N(\beta)) &\le 
\frac{1}{2\ell_{\min}(\beta) + (N 
\eta)^{-1}}, \\
\lambda_{\min,V} \equiv \lambda_{\min}(N S_V(\beta)) &\ge \frac{1}{\ell_{\max}(\beta) + (N 
\eta)^{-1}}, \\ 
\lambda_{\max,V} \equiv \lambda_{\max}(N S_V(\beta)) &\le \frac{1}{\ell_{\min}(\beta) + (N 
\eta)^{-1}}.
\ee

\subsection*{Part (c): Control of KL Divergence}
Now we show control of $Q$, assuming that $\|\Sigma-\Sigma_N\| \le \delta$. We 
suppress dependence on $\beta$ in the following.
\be
Q = &(S_V W'\kappa-S_N W'\kappa )' S_N^{-1}  ( S_V W'\kappa-S_N W'\kappa ) \\
= & \frac{1}{N} (NS_V W'\kappa-NS_N W'\kappa )' (NS_N)^{-1}  ( NS_V W'\kappa-N 
S_N W'\kappa ) \\
&\le \frac{1}{N \lambda_{\min,N}} \norm{N S_V W'\kappa-N S_N W'\kappa}^2 \\
&= \frac{1}{N \lambda_{\min,N}} \norm{N S_V\left[
W'\kappa- N^{-1} S_V^{-1} N S_N W'\kappa \right]}^2 \\
&\le \frac{\lambda_{\max,V}^2}{N \lambda_{\min,N}}~\norm{\left[I-N^{-1} S_V^{-1} 
N S_N \right] W'\kappa}^2 \\
&\le \frac{\lambda_{\max,V}^2}{N \lambda_{\min,N}}~\norm{I-N^{-1} S_V^{-1} N S_N 
}^2 \norm{W'\kappa}^2 \\
&\le \frac{\lambda_{\max,V}^2}{N \lambda_{\min,N}}~\delta^2 
\lambda_{\max,N}^2  \norm{W'\kappa}^2  \\
&\le p \frac{\lambda_{\max,V}^2 \delta^2  
\lambda_{\max,N}^2}{\lambda_{\min,N}}
\ee
where various steps used Cauchy-Schwartz, assume $W$ is standardized to unit 
variance, $\kappa_i \in \{-1/2,1/2\}$, $\|\Sigma_V -\Sigma_N\| < \delta$, and 
\be
&\norm{I-N^{-1} S_V^{-1} N S_N} \le \norm{(N S_N)^{-1} - (N S_V)^{-1}}~\norm{N S_N}.
\ee 

To bound the other terms in the KL, first note that 
\be
\tr{ S_N^{-1} S_V} -p &= \tr{ (NS_N)^{-1} (NS_V) - I} \\
&= \tr{ ( \Sigma_N - \Sigma_V ) N S_V} 
\\
&\le \lambda_{\max}(\Sigma_N-\Sigma_V) \tr{N S_V} \le \frac{2p 
\delta}{\lambda_{\min,N}}.
\ee
Further, from Lemma B.2 in \cite{pati2014posterior}, since $S_V$ and $S_N$ are 
both positive definite for $|V|>p$, $\log |S_N S_V^{-1}| < \tr{S_N^{-1} S_V - 
I}$. So putting all of the bounds together,
\be
\KL{\mc N_V}{\mc N_N} 
&\le \frac{p\lambda_{\max,V}^2 \delta^2  
\lambda_{\max,N}^2}{\lambda_{\min,N}} + 
\frac{4p \delta}{\lambda_{\min,N}} \\
&\le \left( \frac{p\delta^2}{(\lambda^- 
+(N\eta)^{-1})^4 } + 
4p \delta \right) (\lambda^+ + 
\frac{\lambda^-}{2} + (N\eta)^{-1})
\ee
with $\lambda^- = \lambda_{\min}(\Sigma_N)$ and $\lambda^+ = 
\lambda_{\max}(\Sigma_N)$. 

\subsection*{Part (d): Uniform control of $\tv{\Krth}{\Krthe}$}
Putting
\be
\delta = \frac{(\lambda^-+(N\eta)^{-1})^2}{2p  (\lambda^+ + 
\frac{\lambda^-}{2} + (N\eta)^{-1})} \epsilon^2
\ee
gives
\be
\KL{\mc N_V}{\mc N_N} &\le 2 \epsilon^2 \\
\| \mc N_V - \mc N_N \|_{TV} &\le \epsilon.
\ee
The state dependence arises in two places. First, $\lambda^+$ and $\lambda^-$ 
are functions of $\omega$. Second, the probability \eqref{eq:Conc} of achieving 
this condition depends on the random constant $K$ in \eqref{eq:ConcCons}, which 
also depends on $\omega$.

\section{Proof of Remark \ref{rem:LogregCons}}
The result is a consequence of the following Lemma
\begin{lemma} \label{lem:pgexptail}
The $\textnormal{PG}(1,\alpha)$ distribution is a log-concave probability law.
\end{lemma}
\begin{proof}
If $\omega \sim \mbox{PG}(1,\alpha)$, then it is equal in distribution to the 
infinite sum of Exponentials 
\begin{align*}
\omega \sim \sum_{k=0}^{\infty} \varphi_k,\quad 
\varphi_k = \frac{g_k}{\pi^2 (k-1/2)^2 + \alpha^2/2},
\end{align*}
where $g_k \sim \mbox{Exp}(1)$, $\varphi_k \sim \mbox{Exp}(\pi^2 (k-1/2)^2 + 
\alpha^2/2)$, 
and $\varphi_k$ has a log-concave probability distribution since 
$\mbox{Exp}(\lambda)$ is log-concave for all finite $\lambda$ (see e.g. 
\cite{bagnoli2005log}).
Consider the sequence of random variables
\begin{align*}
\omega_n \sim \sum_{k=0}^{n} \frac{g_k}{\pi^2 (k-1/2)^2 + \alpha^2/2} = 
\sum_{k=0}^{\infty} \varphi_k
\end{align*}
for $n=0,\ldots,\infty$. For any finite $n$, $\omega_n$ has a log-concave 
distribution since the sum of independent random variables having log-concave 
distributions is log-concave (see Proposition 3.5 in \cite{saumard2014log}).  
As 
$\omega_n \stackrel{D}{\to} \omega$ (indicating convergence in distribution), 
$\omega$ is log concave from Proposition 3.6 in \cite{saumard2014log}.
\end{proof}
Since $\omega_i$ has a density, it follows that this density is log-concave, 
and therefore so is its survival function. Therefore we can write the survival 
function as $\Pp[\omega_i > t] = e^{-V_i(t)}$ for some convex function $V_i$ and
\be
\Pp \left[ \bigvee_i \omega_i > t \right] &= 1-\Pp \left[ \bigvee_i \omega_i < 
t \right] = 1- \prod_i (1-e^{-V_i(t)}).
\ee
The other claims of Remark \ref{rem:LogregCons} are established in the proof of 
Theorem \ref{thm:logregerr2}.

\section{Proof of Theorem \ref{thm:gperror}}
The results in this section concern the model in \eqref{eq:npreg} with priors 
in \eqref{eq:gpprior}. The transition kernel $\mc P$ is induced by the
sampler defined in section \ref{sec:gp}, and the approximating kernel $\mc 
P_{\epsilon}$ substitutes $\tSigma=U_{\epsilon} \Lambda_{\epsilon} 
U_{\epsilon}'$ for $\Sigma$ where $U_{\epsilon}$ is $N \times k$, 
$\Lambda_{\epsilon}$ is $k \times k$, and $k \le N$.

The result uses the following Lemma.
\begin{lemma} \label{lem:mhapprox}
Suppose $\P$ is the transition kernel of a Metropolis-Hastings algorithm
with proposal kernel $Q(x,dy)$ and acceptance probability $\alpha(x,y)$ 
and $\P_\epsilon$ uses the same proposal kernel but has acceptance
probability $\alpha_\epsilon(x,y)$. If
\be \label{eq:SupAccRat}
\sup_{x,y \in \X \times \X} |\alpha(x,y) - \alpha_\epsilon(x,y)| < \frac{\epsilon}2
\ee
then
\be
\sup_{x \in \X} \| \P(x,\cdot)-\P_\epsilon(x,\cdot)\|_{TV} \le \epsilon.
\ee
\end{lemma}
\begin{proof}
We have
\be
\sup_{x \in \X} \sup_{|\phi|<1} &\int \phi(y) \alpha(x,y) Q(x,dy) + \int \phi(x) (1-\alpha(x,y)) Q(x,dy) \\
&- \int \phi(y) \alpha_\epsilon(x,y) Q(x,dy) + \int \phi(x) (1-\alpha_\epsilon(x,y)) Q(x,dy) \\
&\le \sup_{x \in \X} \int |\phi(y)-\phi(x)| |\alpha(x,y)-\alpha_\epsilon(x,y)| Q(x,dy) \\
&\le 2 \sup_{x,y \in \X \times \X} |\alpha(x,y) - \alpha_\epsilon(x,y) | \le 2 \frac{\epsilon}2 = \epsilon.
\ee
\end{proof}

\subsubsection{Main result: approximation error for GP MH steps}
We now show that for every $\epsilon > 0$, the kernel $\mc P_{\epsilon}$ that 
replaces $\Sigma$ with $\tSigma$, achieving $\norm{\Sigma -\tSigma} < \delta$, 
satisfies Assumption \ref{ass:UniformTV}.

Applying Lemma \ref{lem:mhapprox}, we need only control \eqref{eq:SupAccRat}
\be
D_{\epsilon}(x,y) &\equiv |\alpha_\epsilon(x,y)-\alpha(x,y)| \\
&= \left| \left(\frac{L_{\epsilon}(z \mid y) p(y) Q(y,x)  }{ L_{\epsilon}(z \mid x) p(x) Q(x,y)} 
\wedge 1 \right) - \left(\frac{L(z \mid y) p(y) Q(y,x)  }{ L(z \mid x) p(x) Q(x,y)} \wedge 
1\right) \right| \\
&= \left| \left(r_{\epsilon}(x,y) \wedge 1 \right) - 
\left(r(x,y) \wedge 1\right) \right|.
\ee
Initially focus on the case where both $r_{\epsilon}(x,y)$ and 
$r(x,y)$ are less than one, and set $M(x,y) = 
\frac{p(y) Q(y,x)}{ p(x) Q(x,y)}$. 
Then 
\be
D_{\epsilon}(x,y)  &=  
M(x,y) \left| \frac{L_{\epsilon}(z \mid y)  }{ 
L_{\epsilon}(z \mid x)} - \frac{L(z \mid y)  }{ L(z \mid x)} 
\right| \\
&= M(x,y) \bigg| \frac{  | \Psi_{\epsilon,y}|^{1/2} }{| \Psi_{\epsilon,x}|^{1/2}} e^{-z' (\Psi_{\epsilon,y}-\Psi_{\epsilon,x}) z/2} 
-\frac{  | \Psi_y|^{1/2} }{| \Psi_x|^{1/2}} e^{-z' (\Psi_y-\Psi_x) z/2} \bigg|.
\ee
Now use that $\Sigma_{\epsilon}$ is a rank $k$ partial eigendecomposition of 
$\Sigma$ satisfying $\|\Sigma_{\epsilon} - \Sigma\|_F < \delta$, implying the 
following
\be
\Psi_{\epsilon}^{-1} &= U (\tau^2 \Lambda_{\epsilon} + 
\sigma^2 I) U', \qquad \Psi^{-1} = U (\tau^2 \Lambda + \sigma^2 
I) U' \\
\lambda^{\epsilon}_i &= \lambda_i, i \le k, \qquad \lambda^{\epsilon}_i = 0, i 
> k, \qquad \lambda_i < \delta, i > k,
\ee
where $\Lambda_\epsilon = \mbox{diag}(\lambda_1,\ldots,\lambda_k,0,\ldots,0)$, 
and $\lambda_i$ is the $i$th eigenvalue of $\Sigma$. Now put $z_U = z' U$, with 
$i$th entry $z_{U,i}$, $\Phi_\epsilon = (\tau^2 \Lambda_{\epsilon} + 
\sigma^2 I)^{-1}$, $\Phi = (\tau^2 \Lambda + 
\sigma^2 I)^{-1}$, and obtain
\be
D_{\epsilon}(x,y) &= M(x,y) \bigg| e^{-\frac{1}{2} z_U (\Phi_{\epsilon,y}-\Phi_{\epsilon,x}) z_U' + \frac12 \log \frac{  | \Phi_{\epsilon,y}| }{| \Phi_{\epsilon,x}|}}  
-e^{-\frac{1}{2} z_U (\Phi_y-\Phi_x) z_U' + \frac12 \log \frac{  | \Phi_y| }{| \Phi_x|}} \bigg| \\
\ee
Define
\be
g(x,y,z,\Phi) = e^{-\frac{1}{2} z (\Phi_{y}-\Phi_{x}) z' + \frac12 \log \frac{  | \Phi_y| }{| \Phi_x|}}  
\ee
giving
\be
D_{\epsilon}(x,y) &= M(x,y) g(x,y,z_U^{(1:k)},\Phi^{(1:k)}) 
|g(x,y,z_U^{(k+1:N)},\Phi_\epsilon^{(k+1:N)}) - 
g(x,y,z_U^{(k+1:N)},\Phi^{(k+1:N)})| \\
&= M(x,y) g(x,y,z_U^{(1:k)},\Phi^{(1:k)}) \Delta(\delta,x,y)
\ee
where for a vector $z$, $z^{(j:k)}$ is the subvector consisting of the  
$j$th through $k$th elements of $z$, and for a matrix $\Psi$, $\Psi^{(j:k)}$ is 
the submatrix consisting of elements with both row and column indices in the 
range $j:k$. Simplifying gives
\be
\Delta(\delta,x,y) &= \bigg| e^{-\frac{1}{2} \sum_{i=k+1}^N 
z_{U,i}^2 \left[ \frac{1}{\sigma_y^2} - \frac{1}{\sigma_x^2} \right]  + 
\frac{1}{2} \sum_{i=k+1}^N \log\frac{\sigma_y^2}{\sigma^2_x} }   \\
&-e^{-\frac{1}{2} \sum_{i=k+1}^N z_{U,i}^2 \left[ \frac{1}{\tau_y^2 
\lambda_i + \sigma_y^2} - \frac{1}{\tau^2_x \lambda_i + \sigma^2_x} \right] 
+ \frac{1}{2} \sum_{i=k+1}^N \log\frac{\tau_y^2 \lambda_i + 
\sigma_y^2}{\tau^2_x 
\lambda_i + \sigma^2_x} }  \bigg| \\
&\le \bigg| e^{-\frac{N-k}{2} \left[ 
\frac{\sigma^2_x-\sigma_y^2}{\sigma_y^2 \sigma^2_x}  - 
\log\frac{\sigma_y^2}{\sigma^2_x} \right] } -e^{-\frac{N-k}{2} \left[ 
\frac{\tau_x^2 \delta + \sigma_x^2 - \tau_y^2 
\delta - \sigma_y^2}{(\tau_y^2 \delta + \sigma_y^2)(\tau^2_x \delta + 
\sigma^2_x)} - \log\frac{\tau_y^2 \delta + \sigma_y^2}{\tau^2_x \delta + 
\sigma^2_x} \right] }  \bigg| \\
&= \bigO{\delta}
\ee

Although the only case considered above was that where 
$r_{\epsilon}(x,y) < 1$, if $|r_{\epsilon}(x,y) - r(x,y)| < 
\frac{\epsilon}{2}$, then
\begin{align*}
| (1 \wedge r_{\epsilon}(x,y)) - (1 \wedge r_{\epsilon}(x,y)) | < 
\frac{\epsilon}{2}.
\end{align*}
concluding the proof.

\end{appendix}

\bibliographystyle{plainnat}
\bibliography{amcmc-AoS}

\begin{thebibliography}{34}
\providecommand{\natexlab}[1]{#1}
\providecommand{\url}[1]{\texttt{#1}}
\expandafter\ifx\csname urlstyle\endcsname\relax
  \providecommand{\doi}[1]{doi: #1}\else
  \providecommand{\doi}{doi: \begingroup \urlstyle{rm}\Url}\fi

\bibitem[Alquier et~al.(2014)Alquier, Friel, Everitt, and
  Boland]{alquier2014noisy}
Pierre Alquier, Nial Friel, Richard Everitt, and Aidan Boland.
\newblock Noisy {M}onte {C}arlo: {C}onvergence of {M}arkov chains with
  approximate transition kernels.
\newblock \emph{Statistics and Computing}, 25\penalty0 (1):\penalty0 1--19,
  2014.

\bibitem[Bagnoli and Bergstrom(2005)]{bagnoli2005log}
Mark Bagnoli and Ted Bergstrom.
\newblock Log-concave probability and its applications.
\newblock \emph{Economic Theory}, 26\penalty0 (2):\penalty0 445--469, 2005.

\bibitem[Banerjee et~al.(2013)Banerjee, Dunson, and
  Tokdar]{banerjee2012efficient}
Anjishnu Banerjee, David~B. Dunson, and Surya~T. Tokdar.
\newblock {Efficient Gaussian process regression for large datasets}.
\newblock \emph{Biometrika}, 100\penalty0 (1):\penalty0 75--89, 2013.

\bibitem[Banerjee et~al.(2008)Banerjee, Gelfand, Finley, and
  Sang]{banerjee2008gaussian}
Sudipto Banerjee, Alan~E. Gelfand, Andrew~O. Finley, and Huiyan Sang.
\newblock Gaussian predictive process models for large spatial data sets.
\newblock \emph{Journal of the Royal Statistical Society: Series B (Statistical
  Methodology)}, 70\penalty0 (4):\penalty0 825--848, 2008.

\bibitem[Bardenet et~al.(2014)Bardenet, Doucet, and
  Holmes]{bardenet2014towards}
R{\'e}mi Bardenet, Arnaud Doucet, and Chris Holmes.
\newblock Towards scaling up {M}arkov chain {M}onte {C}arlo: an adaptive
  subsampling approach.
\newblock In \emph{Proceedings of the 31st International Conference on Machine
  Learning (ICML-14)}, pages 405--413, 2014.

\bibitem[Bhatia(2013)]{bhatia2013matrix}
Rajendra Bhatia.
\newblock \emph{Matrix analysis}, volume 169.
\newblock Springer Science \& Business Media, 2013.

\bibitem[Bhattacharya and Dunson(2010)]{bhattacharya2010nonparametric}
Abhishek Bhattacharya and David~B. Dunson.
\newblock {Nonparametric Bayesian density estimation on manifolds with
  applications to planar shapes}.
\newblock \emph{Biometrika}, 102\penalty0 (2):\penalty0 851--865, 2010.

\bibitem[Choi and Hobert(2013)]{choi2013polya}
Hee~Min Choi and James~P Hobert.
\newblock {The P\'{o}lya-Gamma Gibbs sampler for Bayesian logistic regression
  is uniformly ergodic}.
\newblock \emph{Electronic Journal of Statistics}, 7:\penalty0 2054--2064,
  2013.

\bibitem[Ferr{\'e} et~al.(2013)Ferr{\'e}, Herv{\'e}, and
  Ledoux]{ferre2013regular}
D{\'e}borah Ferr{\'e}, Lo{\"\i}c Herv{\'e}, and James Ledoux.
\newblock Regular perturbation of {V}-geometrically ergodic {M}arkov chains.
\newblock \emph{Journal of Applied Probability}, 50\penalty0 (1):\penalty0
  184--194, 2013.

\bibitem[Finley et~al.(2009)Finley, Sang, Banerjee, and
  Gelfand]{finley2009improving}
Andrew~O. Finley, Huiyan Sang, Sudipto Banerjee, and Alan~E. Gelfand.
\newblock Improving the performance of predictive process modeling for large
  datasets.
\newblock \emph{Computational Statistics \& Data Analysis}, 53\penalty0
  (8):\penalty0 2873--2884, 2009.

\bibitem[Gamerman and Lopes(2006)]{gamerman2006markov}
Dani Gamerman and Hedibert~F. Lopes.
\newblock \emph{Markov chain Monte Carlo: stochastic simulation for Bayesian
  inference}.
\newblock CRC Press, 2 edition, 2006.

\bibitem[Guhaniyogi et~al.(2014)Guhaniyogi, Qamar, and
  Dunson]{guhaniyogi2014bayesian}
Rajarshi Guhaniyogi, Shaan Qamar, and David~B. Dunson.
\newblock Bayesian conditional density filtering.
\newblock \emph{arXiv preprint arXiv:1401.3632}, 2014.

\bibitem[Halko et~al.(2011)Halko, Martinsson, and Tropp]{halko2011finding}
Nathan Halko, Per-Gunnar Martinsson, and Joel~A. Tropp.
\newblock Finding structure with randomness: Probabilistic algorithms for
  constructing approximate matrix decompositions.
\newblock \emph{SIAM review}, 53\penalty0 (2):\penalty0 217--288, 2011.

\bibitem[Hoffman and Wielandt(1953)]{hoffman1953variation}
Alan~J. Hoffman and Helmut~W. Wielandt.
\newblock The variation of the spectrum of a normal matrix.
\newblock \emph{Duke Math. J}, 20\penalty0 (1):\penalty0 37--39, 1953.

\bibitem[Hughes and Haran(2013)]{hughes2013dimension}
John Hughes and Murali Haran.
\newblock Dimension reduction and alleviation of confounding for spatial
  generalized linear mixed models.
\newblock \emph{Journal of the Royal Statistical Society: Series B (Statistical
  Methodology)}, 75\penalty0 (1):\penalty0 139--159, 2013.

\bibitem[Johndrow and Mattingly(2017)]{johndrow2017coupling}
James~E. Johndrow and Jonathan~C. Mattingly.
\newblock Coupling and decoupling to bound an approximating {M}arkov chain.
\newblock \emph{arXiv preprint arXiv:1706.02040}, 2017.

\bibitem[Korattikara et~al.(2014)Korattikara, Chen, and
  Welling]{korattikara2014austerity}
Anoop Korattikara, Yutian Chen, and Max Welling.
\newblock Austerity in {MCMC} land: Cutting the {M}etropolis-{H}astings budget.
\newblock In \emph{Proceedings of the 31st International Conference on Machine
  Learning (ICML-14)}, pages 181--189, 2014.

\bibitem[Mitrophanov(2005)]{mitrophanov2005sensitivity}
Alexander~Y. Mitrophanov.
\newblock Sensitivity and convergence of uniformly ergodic {M}arkov chains.
\newblock \emph{Journal of Applied Probability}, 42\penalty0 (4):\penalty0
  1003--1014, 2005.

\bibitem[Nash et~al.(1994)Nash, Sellers, Talbot, Cawthorn, and
  Ford]{nash1994population}
Warwick~J. Nash, Tracy~L. Sellers, Simon~R. Talbot, Andrew~J. Cawthorn, and
  Wes~B. Ford.
\newblock The population biology of {A}balone ({H}aliotis species).
\newblock \emph{Blacklip Abalone (H. rubra) from the North Coast and Islands of
  Bass Strait. Sea Fisheries Division Technical Report}, 48, 1994.

\bibitem[O'Brien and Dunson(2004)]{o2004bayesian}
Sean~M. O'Brien and David~B. Dunson.
\newblock Bayesian multivariate logistic regression.
\newblock \emph{Biometrics}, 60\penalty0 (3):\penalty0 739--746, 2004.

\bibitem[Pati et~al.(2014)Pati, Bhattacharya, Pillai, and
  Dunson]{pati2014posterior}
Debdeep Pati, Anirban Bhattacharya, Natesh~S. Pillai, and David~B. Dunson.
\newblock Posterior contraction in sparse {B}ayesian factor models for massive
  covariance matrices.
\newblock \emph{The Annals of Statistics}, 42\penalty0 (3):\penalty0
  1102--1130, 2014.

\bibitem[Pillai and Smith(2015)]{pillai2014ergodicity}
Natesh~S. Pillai and Aaron Smith.
\newblock Ergodicity of approximate {MCMC} chains with applications to large
  data sets.
\newblock \emph{arXiv preprint arXiv:1405.0182v2}, 2015.

\bibitem[Polson et~al.(2013)Polson, Scott, and Windle]{polson2013bayesian}
Nicholas~G. Polson, James~G. Scott, and Jesse Windle.
\newblock {Bayesian inference for logistic models using P\'{o}lya--Gamma latent
  variables}.
\newblock \emph{Journal of the American Statistical Association}, 108\penalty0
  (504):\penalty0 1339--1349, 2013.

\bibitem[Quiroz et~al.(2014)Quiroz, Villani, and Kohn]{quiroz2014speeding}
Matias Quiroz, Mattias Villani, and Robert Kohn.
\newblock Speeding up {MCMC} by efficient data subsampling.
\newblock \emph{arXiv preprint arXiv:1404.4178}, 2014.

\bibitem[Ritter and Tanner(1992)]{ritter92}
Christian Ritter and Martin~A. Tanner.
\newblock Facilitating the {G}ibbs sampler: The {G}ibbs stopper and
  griddy-{G}ibbs sampler.
\newblock \emph{Journal of the American Statistical Association}, 87:\penalty0
  861--868, 1992.

\bibitem[Robert and Casella(2004)]{robert2004monte}
Christian~P. Robert and George Casella.
\newblock \emph{Monte {C}arlo statistical methods}.
\newblock Springer, 2 edition, 2004.

\bibitem[Roberts et~al.(1998)Roberts, Rosenthal, and
  Schwartz]{roberts1998convergence}
Gareth~O. Roberts, Jeffrey~S. Rosenthal, and Peter~O. Schwartz.
\newblock Convergence properties of perturbed {M}arkov chains.
\newblock \emph{Journal of Applied Probability}, 35\penalty0 (1):\penalty0
  1--11, 1998.

\bibitem[Rudolf and Schweizer(2017)]{rudolf2015perturbation}
Daniel Rudolf and Nikolaus Schweizer.
\newblock Perturbation theory for markov chains via wasserstein distance.
\newblock \emph{Bernoulli}, 2017.

\bibitem[Sariyar et~al.(2011)Sariyar, Borg, and
  Pommerening]{sariyar2011controlling}
Murat Sariyar, Andreas Borg, and Klaus Pommerening.
\newblock Controlling false match rates in record linkage using extreme value
  theory.
\newblock \emph{Journal of Biomedical Informatics}, 44\penalty0 (4):\penalty0
  648--654, 2011.

\bibitem[Saumard and Wellner(2014)]{saumard2014log}
Adrien Saumard and Jon~A. Wellner.
\newblock {Log-concavity and strong log-concavity: A review}.
\newblock \emph{Statistics Surveys}, 8:\penalty0 45--114, 2014.

\bibitem[Smola and Bartlett(2001)]{smola2001sparse}
Alex~J. Smola and Peter Bartlett.
\newblock Sparse greedy {G}aussian process regression.
\newblock In \emph{Advances in Neural Information Processing Systems 13}, 2001.

\bibitem[Tao(2015)]{tao2015eigenvalues}
Terence Tao.
\newblock 254a, notes 3a: Eigenvalues and sums of {H}ermitian matrices.
\newblock
  \url{https://terrytao.wordpress.com/2010/01/12/254a-notes-3a-eigenvalues-and-sums-of-hermitian-matrices/\#more-3341},
  2015.
\newblock Accessed: 2015-11-02.

\bibitem[Tropp(2012)]{tropp2012user}
Joel~A. Tropp.
\newblock User-friendly tail bounds for sums of random matrices.
\newblock \emph{Foundations of Computational Mathematics}, 12\penalty0
  (4):\penalty0 389--434, 2012.

\bibitem[Waugh(1995)]{waugh1995extending}
Sam Waugh.
\newblock Extending and benchmarking cascade-correlation.
\newblock \emph{Dept of Computer Science, University of Tasmania, Ph. D.
  Dissertation}, 1995.

\end{thebibliography}

%\vfill\hfill{\tiny Last edited: \today, \currenttime~EDT}
\end{document}